\documentclass[11pt]{article}
\usepackage{graphicx}
\usepackage{amstext}
\usepackage{amsmath}
\usepackage{amssymb}
\usepackage{amsbsy}
\usepackage{colordvi}

\newtheorem{theorem}{Theorem}
\newtheorem{lemma}{Lemma}

\newtheorem{corollary}{Corollary}
\newtheorem{definition}{Definition}

\newtheorem{remark}{Remark}

\newenvironment{proof}%
               {\par\noindent%
               \setlength{\parindent}{0em}%
               \setlength{\parskip}{1ex}%
               \textit{Proof:\nopagebreak[2]
}}{\hfill\rule{1.5ex}{1.5ex}\par}

\newcommand{\beq}{\begin{equation}}
\newcommand{\eeq}{\end{equation}}
\newcommand{\beqa}{\begin{eqnarray}}
\newcommand{\eeqa}{\end{eqnarray}}

\newcommand {\mc}{\mathcal}
\newcommand {\mb}{\mathbf}

\newcommand{\n}{\mathcal{N}}

\newcommand{\dfn}{{\triangleq}}

\newcommand{\bx}{\mathbf{x}}
\newcommand{\btx}{\tilde{\mathbf{x}}}
\newcommand{\by}{\mathbf{y}}

\newcommand{\bz}{\mathbf{z}}

\newcommand{\Rea}{\mathbb{R}}
\newcommand{\intgrs}{\mathbb{Z}}
\newcommand{\rationals}{\mathbb{Q}}
\newcommand{\diag}{\text{diag}}
\newcommand{\creg}{\mathcal{C}}
\begin{document}
\allowdisplaybreaks
\title{On the Degrees-of-Freedom of the $K$-User Gaussian Interference Channel}

\author{Raul Etkin and Erik Ordentlich\\ HP Laboratories\\ Palo Alto, CA}

\maketitle
\begin{abstract}
The degrees-of-freedom of a $K$-user Gaussian interference channel
(GIFC) has been defined to be the multiple of $ (1/2)\log_2 P $ at which the
maximum sum of achievable rates grows with increasing $ P $.  In this paper, we
establish that the degrees-of-freedom of three or more user, real,
scalar GIFCs, viewed as 
a function of the channel coefficients, is discontinuous at points
where all of the coefficients are non-zero rational numbers.  More
specifically, for all $ K > 2$, we find a class of $ K $-user GIFCs
that is dense in the 
GIFC parameter space for which $ K/2 $ degrees-of-freedom are exactly
achievable, and we show that the degrees-of-freedom for any GIFC with
non-zero rational coefficients is strictly smaller than $K/2$.  These results
are proved using new connections with number theory and additive
combinatorics. 
\end{abstract}

\section{Introduction}
The time-invariant, real, scalar $ K $-user Gaussian interference
channel (GIFC), as introduced in~\cite{Car78}, involves $ K $
transmitter-receiver pairs in 
which each transmitter attempts to
communicate a uniformly distributed, finite-valued message to its
corresponding receiver by 
sending a signal comprised of $ n $ real numbers.
Each receiver observes a component-wise linear combination of possibly
{\em all} of the transmitted signals plus additive memoryless Gaussian
noise, and seeks to decode, with probability close to one, the message
of its corresponding transmitter, in spite of the interfering signals
and noise.  The time averages of the squares of the transmitted signal values
are required to not exceed certain power constraints.  A $K$-tuple of
rates $ (R_1,\ldots, R_K) $ is said to be achievable for a GIFC if the
transmitters can increase the sizes of their message sets as $
2^{nR_i} $ with the signal length $ n $, and signal in such a way that
the power constraints are met and the receivers are able to correctly
decode their corresponding messages with probability converging to $ 1
$, as $ n $ grows to infinity.  The set of all achievable $K$-tuples of
rates is known as 
the capacity region of the GIFC.  Determining it, as a
function of the channel coefficients (specifying the 
linear combinations mentioned above), power constraints, and noise
variances, has been an open problem in information theory for over 30 years.  

A complete solution for even the two-user case, which has received the
 most attention to date, is still out of reach.  The best known
 coding scheme for two users is that presented in~\cite{HK81}.
 In some ranges of
channel coefficients, such as for strong interference, the
 capacity region is 
completely known for two users~\cite{Car78}.  Still for other ranges,
 the maximum achievable sum-of-rates is known~\cite{Sas04, Sha+07,
 MotKha08, AnnVee08}. For the general two-user case, the strongest
 known result is that of~\cite{Etk+07}, which determines the capacity region
 to within a $ 1/2 $ bit margin ($1$ bit for the complex case) using a
 carefully chosen version of the scheme of~\cite{HK81}, and a new
 genie-aided outerbound. 

The case of $ K > 2 $ users has, until very recently, received
less attention.  Much of the recent effort on $ K > 2 $,
beginning with~\cite{HN05} and continuing in e.g.,~\cite{CadJaf07,CJS07}
has focused on characterizing the growth of the capacity region in the
limit of increasing signal-to-noise ratio (SNR) corresponding, for
example, to fixing the noise variances and channel coefficients and
letting the power constraints tend to infinity.  Specific attention
has been directed at the growth of the maximum sum of achievable rates.
If there were no interference, the maximum achievable rate
corresponding to each transmitter-receiver pair would grow like $
(1/2)\log_2 P $ in the limit of increasing power, which follows from the
well known formula of $ (1/2)\log_2(1+P/N) $ for the capacity of a single
user additive Gaussian noise channel with power constraint $ P $ and
noise variance $ N $.  Thus, the maximum sum of achievable rates would grow
as $ (K/2)\log_2 P $ if there were no interference.  This motivates the
expectation that, in the general
case, the maximum sum of achievable rates would grow
as $ (d/2) \log_2 P $ for some constant $ d \leq K $, depending on the channel
coefficients, where $ d $ has been dubbed the degrees-of-freedom of
the underlying GIFC.  
Although determining $ d $ for a given GIFC is,
in principle, simpler than determining the capacity region,
it has turned out to be a difficult problem in its own right, for $ K
> 2 $.\footnote{The degrees-of-freedom is known to be 1 for all two-user GIFCs, unless there is no interference~\cite{HN05}.} 

A positive development in the study of the degrees-of-freedom of
GIFCs with more than two users has been the discovery of a new coding
technique known
as {\em interference alignment}, which involves carefully choosing the
transmitted signals so that the interfering signals ``align'' benignly at
each receiver~\cite{CadJaf07}.  Interference alignment has been shown,
under some 
conditions which we summarize below, to achieve nearly $ d = K/2 $
degrees-of-freedom, which is half of the degrees-of-freedom in the case
of no interference at all.
Interference alignment is not possible to implement for two users and its discovery thus had to wait until the focus shifted to more
users.  Another new phenomenon in network
information theory that has recently emerged as the number of users studied
was increased, is the technique of {\em indirect decoding}, which is
crucial for achieving the capacity region of certain three-user broadcast
channels~\cite{NaiElG08}.  Again, this technique is not relevant in the
two-user case, and could not have been discovered in the study thereof.

In this paper, we find a new information theoretic
phenomenon concerning interference channels that is not manifest in
the two-user case.  In particular, we find that the
degrees-of-freedom (and therefore the capacity region at high
signal-to-noise ratio) of real, scalar GIFCs with $ K > 2 $ users is very
sensitive to whether the channel coefficients determining the linear
combinations of signals at each receiver are rational or irrational numbers.  
Next, we formally explain our results and their significance in the
context of the growing literature on $ K > 2 $ user GIFCs. 

We shall use a matrix $ H $ to denote the direct and cross gains of a
time invariant, real, scalar $ K $-user (GIFC)~\cite{Car78} with the $
(i,j) $-th entry $ h_{i,j} $ specifying the channel gain 
from transmitter $ i $ to receiver $ j $.  Thus, the signal
observed by receiver $ j \in \{1,\ldots,K\} $ at time index $ t = 1,2,\ldots $
is given by  $ z_{j,t}+\sum_{i=1}^K x_{i,t}h_{i,j} $ where 
$ x_{i,t} $ is the real valued signal of transmitter $ i \in
\{1,\ldots,K\} $ at time $ t $ and $ z_{j,t} $
is additive Gaussian noise with variance $ \sigma^{2}_j $, independent
across time and users.  Fixing a block length $ n $, the transmitted
signals $\{x_{i,t}\} $ are required to satisfy the average power constraints $
\sum_{t=1}^n x_{i,t}^2 \leq nP_{i} $ for some collection of powers $
P_1,\ldots,P_K $.
For $ H \in \Rea^{K\times
  K} $, and $ \boldsymbol{\sigma}, \mathbf{P} \in \Rea_+^{K} $, we let
$ \creg(H,\boldsymbol{\sigma},\mathbf{P}) $ denote the capacity region
of a GIFC with gain matrix $ H $, receiver noise variances given
by the corresponding components of $ \boldsymbol{\sigma} $, and
average (per codeword) power constraints given by the components of $
\mathbf{P} $. 
Following~\cite{HN05}, we
define the degrees-of-freedom of $ H $ as 
\beq
DoF(H) = \limsup_{P \rightarrow \infty}\frac{\max_{\mathbf{R} \in
  \creg(H,\mathbf{1},P\mathbf{1})} \mathbf{1}^t \mathbf{R}}{(1/2)\log_2 P} ,
\label{eq:dofdef}
\eeq
where $ \mathbf{1} $ denotes the vector of all ones.
The degrees-of-freedom of a GIFC characterizes the behavior of the maximum
achievable sum rate as the SNR tends
to infinity, with the gain matrix fixed.

A fully connected GIFC is
one for which $ h_{i,j} \neq 0 $ for all $ i $ and $ j $.
It was shown in~\cite{HN05} that for fully connected $ H $, $ DoF(H) \leq K/2
$.  
If $ H $ is not fully connected, the
degrees-of-freedom can be as high as $ K $, such as when $ H $ is
the identity matrix where all cross gains are zero.  Little was known
about tightness of the $ K/2 $ bound for $ K > 2 $ until it was shown
in~\cite{CadJaf07} that  
for {\em vector} GIFCs and an appropriate generalization of $ DoF(\cdot) $ to include a  
normalization by the input/output vector dimension, the
degrees-of-freedom of ``almost all'' fully connected vector GIFCs
approaches $ K/2 $ 
when the vector dimension tends to infinity.\footnote{The $ K/2 $
  bound of~\cite{HN05} extends to the fully connected vector case as
  well.}  In addition, an example of a fully connected
two-dimensional vector GIFC achieving exactly $ K/2 $
degrees-of-freedom was also given in~\cite{CadJaf07}.   
The key tool introduced
in~\cite{CadJaf07} to establish these results is the technique of 
  interference alignment, which involves the transmitters signaling over
linear subspaces that, after component-wise scaling by the cross gains, {\em align} into interfering subspaces which are linearly independent with the directly received subspaces, allowing for many interference free dimensions over which to communicate.  For real, scalar GIFCs, it was  
shown in~\cite{CJS07}, using a different type of interference alignment,
that the degrees-of-freedom of certain fully connected GIFCs also
approaches $ K/2 $ when the cross gains
tend to zero.  Yet a different type of interference alignment is used
in~\cite{Bre+07} to find new achievable rates for a non-fully
connected GIFC in which interference occurs only at one receiver.
To our knowledge, the problem of determining or computing the degrees-of-freedom of general GIFCs is still open.  

As in~\cite{CJS07}, in this paper we consider only fully connected scalar,
real GIFCs and establish the following results on the degrees-of-freedom.

\begin{theorem}
\label{thm:irrational} 
If all diagonal components of a fully connected $ H $ are irrational algebraic numbers and all off-diagonal components are rational numbers then $ DoF(H) = K/2 $.
\end{theorem}

\begin{theorem}
\label{thm:rational}
For $ K > 2 $,
if all elements of a fully connected $ H $ are rational numbers then $ DoF(H) < K/2 $.
\end{theorem}

The following corollary is then immediate from
Theorems~\ref{thm:irrational} and~\ref{thm:rational}, and the well
known fact that irrational algebraic numbers are dense in the real
numbers.

\begin{corollary}
\label{cor:discont}
For $ K > 2 $, the function $ DoF(H) $ is discontinuous at all fully connected $ H $ with rational components.
\end{corollary}

Theorem~\ref{thm:irrational} demonstrates the existence of fully
connected, real $ K $-user GIFCs with exactly $ K/2 $ degrees-of-freedom.  In contrast to the result of~\cite{CJS07},
Theorem~\ref{thm:irrational} is non-asymptotic (in $ H $).  The
underlying achievability scheme is based
on an interference alignment phenomenon that differs from the ones used
in~\cite{CadJaf07} and~\cite{CJS07}, and
relies on number theoretic lower bounds on the approximability of
irrational algebraic numbers by rationals.

Theorem~\ref{thm:rational} reveals a surprising limitation on $ 
DoF(H) $ when the components are non-zero rational numbers (up to
arbitrary pre-post multiplication by diagonal matrices - see
Lemma~\ref{lem:invariance} in Section~\ref{sec:irrational}).  In this
case, $ DoF(H) $ is strictly 
bounded away from $ K/2 $.  Previously known techniques for
finding outer-bounds to the capacity regions of GIFCs, such as
cooperative encoding and decoding~\cite{Sat77,Sat78}, genie
aided decoding~\cite{Etk+07,Car83,Kra04}, and multiple access
bounds~\cite{Car78, CadJaf08}
are not sensitive to the rationality of the channel
parameters and hence do not suffice to establish
Theorem~\ref{thm:rational}.  Instead, our proof of
this theorem is based on a new connection between GIFCs
with rational $ H $
and results from additive 
combinatorics~\cite{TaoVu06}, a branch of combinatorics that is
concerned with the cardinalities of sum sets, or sets obtained by
adding (assuming an underlying group 
structure) any element of a set $ A $ to any element of a set $ B $.

The remainder of the paper is organized as follows.  
The next section clarifies some notation and gives the formal
definition of the capacity region of a GIFC that will apply in this
paper. In
Section~\ref{sec:irrational}, we present the proof of
Theorem~\ref{thm:irrational}.  This is followed by the proof of
Theorem~\ref{thm:rational} in Section~\ref{sec:rational}, which
further consists of subsections collecting various intermediate
results.  Each of these sections is prefaced
with a high level outline of the respective proofs.
In Section~\ref{sec:example}, we determine
lower and upper bounds on $ DoF(H) $ for a simple three-user rational $ H $
by improving on the scheme of~\cite{CJS07}, and evaluating an upper bound
implicit in the proof of Theorem~\ref{thm:rational}.  We
conclude in Section~\ref{sec:conclusion} with some final observations
and directions for future work. 

\section{Notation and definitions}
We adopt the usual notation for the information theoretic quantities
of discrete and differential entropy (resp.\ $ H(X) $ and $ h(X) $),
and mutual information ($ I(X;Y) $), which shall all be measured in bits
(i.e.\ involve logarithms to the base two)~\cite{CT06}.  We shall use
the standard notation $ \lceil x \rceil $ and $ \lfloor x \rfloor $ to
respectively denote the smallest integer not smaller than $ x $ and the
greatest integer not larger than $ x $.  The cardinality of a set $
{\cal A}$ shall be denoted as $ |{\cal A}| $.

Next, we review the definition of
the capacity region of a $ K $-user GIFC with power constraints $
\mathbf{P}=(P_1,\ldots, P_K) $ and noise variances $ \boldsymbol{\sigma}=(\sigma^2_1,\ldots,\sigma^2_K)
$ that will apply in this paper.  Fixing a block length $ n $ and
a rate-tuple $ R_1,\ldots,R_K$, the random
message $W_i$ of the $i$-th transmitter is assumed to be
uniformly distributed in the set $ {\cal W}_i \stackrel{\triangle}{=}
\{1,\ldots,2^{\lceil{nR_i}\rceil}\} 
$.  The messages are further assumed to be independent from one user
to the next.
A coding scheme consists of $ K $ encoding functions
$ f_1,\ldots,f_K $ where $ f_i $ maps $ {\cal W}_i $ into the $ n
$-dimensional ball of radius $ \sqrt{nP_i} $ of
real vectors, the components of which specify the signal
value $ x_{i,t} $ that the $i$-th transmitter will send at each time
index.\footnote{For simplicity, in this paper, we formally adopt a
  per codeword average power constraint, as opposed to the more
  conventional {\em expected} average power constraint.  We note that
  the degrees-of-freedom of a GIFC can be shown to be the same under
  both types of power constraints.}  The set
$\{f_i(1),f_i(2),\ldots,f_i(2^{\lceil{nR_i}\rceil})\}$ constitutes the
codebook of transmitter $ i $.
There is also a corresponding set of $ K $ decoding functions $
g_1,\ldots,g_K $ where $ g_i $ maps $ n $-dimensional real vectors
into the message set $ {\cal W}_i $.  The function $ g_i $ is applied
by receiver $ i $ to the $n$ received signal values $ \mathbf{y}_i^n = y_{i,1},\ldots,y_{i,n} $,
which, as specified in the introduction, are formed as a component-wise linear
combination, according
the gain matrix $ H $,  of the transmitted signals and Gaussian noise.
\begin{definition}
\label{def:capreg}
The capacity region 
$ \creg(H,\boldsymbol{\sigma},\mathbf{P}) $ 
of the GIFC is defined as the set of
rate-tuples $R_1,\ldots,R_K$ for which there exists a sequence of
block length $ n $ message sets and power-constrained coding
schemes satisfying 
$ \lim_{n
  \rightarrow \infty} \max_{1\leq i \leq K} Pr(W_i \neq g_i(\mathbf{y}^{n})) 
= 0 $, where the probability of error is taken with respect to the distribution
induced by the random messages, the coding scheme, and the channel,
as specified above.
\end{definition}

The degrees-of-freedom $DoF(H)$ of a GIFC, as defined
in~(\ref{eq:dofdef}) above, will be assumed to be based on this formal
definition of $ \creg(H,\boldsymbol{\sigma},\mathbf{P}) $.

\section{Real, scalar GIFCs with exactly $ K/2 $ degrees-of-freedom} 
\label{sec:irrational}
In this section, we 
prove Theorem~\ref{thm:irrational} demonstrating the existence of
fully connected, real, scalar $ K $-user GIFCs with exactly $ K/2 $ degrees-of-freedom.  An outline of the proof is as follows.  First, we prove a simple
lemma (which will also be useful in the next section) showing that $
DoF(H) = DoF(D_tHD_r) $ for any diagonal matrices $ D_t $
and $ D_r $ with positive diagonal components.  This, in turn, implies
that we can transform any $ H $ satisfying 
the assumptions of Theorem~\ref{thm:irrational} to one with
irrational, algebraic numbers along the diagonal and integer values in
off-diagonal components, while preserving the degrees-of-freedom.  We
then focus on coding schemes for the new $ H $ in which each
transmitter is restricted to signaling over the
scalar lattice $ \{zP^{1/4+\epsilon}: z\in \intgrs\} $ intersected with
the interval $ [-P^{1/2},P^{1/2}] $.  The idea is that the integer
valued cross gains 
guarantee that the interfering signal values at each receiver will also be
confined to this scalar lattice (though may fall outside of the $
P^{1/2} $ interval), while the irrational direct gains
place the directly transmitted signal values on a scaled lattice that
``stands out'' from the interfering lattice.  Specifically, this
scaled lattice has the property that offsetting the interfering
lattice (equal to the 
original lattice) by each point in the scaled lattice results in 
disjoint sets.  A non-empty intersection would imply that the 
direct gain could be written as the ratio of two integers, which would
contradict its irrationality.  An even stronger property holds for
{\em algebraic} irrational direct gains: the distance between any
pair of points obtained by 
adding a point from the scaled lattice to a point from the 
interfering lattice actually grows with $ P $.  This is shown to follow
from a major result in number theory stating that for
any irrational algebraic number $ \alpha $ and any $ \gamma > 0 $, a
rational $ p/q $ approximation will have an error of at least $
\delta/q^{2+\gamma} $ for some $ \delta $ depending only on $ \alpha
$ and $ \gamma $.\footnote{For irrational algebraic numbers of degree two
(solutions to quadratic equations with integer coefficients), such as $
  \sqrt{2} $, the approximation bound
holds with $ \gamma = 0 $ and is known as 
Liouville's Theorem (established in 1844).  The validity of the bound for 
general algebraic numbers was a longstanding open problem in number
theory and was finally established in 1955 by K.~F.~Roth, for which he
was awarded the Fields Medal.}  The next step in the proof is to deal
with the noise by
coupling this inter-point distance growth with Fano's inequality to show
that the mutual information induced between each transmitter-receiver
pair by independent, uniform distributions on the original
power-constrained lattices, taking
interference into account, grows like $
(1/4-\epsilon) \log_2 P $, for arbitrarily small $ \epsilon $.  This,
in turn, implies the existence of a sequence of block codes (with
symbols from the original lattice) with sum rate approaching $
(K/4)\log_2 P $, and which are correctly decodeable, with high probability, by
treating interference as noise.

As mentioned, we begin with an
invariance property of $ DoF(H) $.
\begin{lemma}[Invariance property]
\label{lem:invariance}
For any matrix $ H $ and diagonal matrices $ D_t $ and $ D_r $ with
positive diagonal components $ DoF(D_tHD_r) = DoF(H)
$.\footnote{The matrices $D_t $ and $ D_r $ need only have non-zero
  diagonal components for the result to hold.  We assume positivity
  for simplicity, as this is all we shall require in this paper.}
\end{lemma}
\begin{proof}
Let $D_t=\diag(d_{t1},\ldots,d_{tK})$ $D_r=\diag(d_{r1},\ldots,d_{rK})$. In the matrix multiplication $D_tHD_r$, $d_{ti}$ scales the channel gains from transmitter $i$ to the different receivers, while $d_{rj}$ scales the channel gains from all transmitters to receiver $j$. By scaling the input signals and noise variances instead of the channel gains, we can write
\beq
\creg(D_tHD_r, \mb{1}, P\mb{1}) = \creg(H, \mb{1}^t D_r^{-2}, P
\mb{1}^t D_t^2). 
\label{eq:eq1}
\eeq

Let $\check {d}_{t} = \min_{1\le i \le K} d_{ti} $, $\hat{d}_{t} =
\max_{1\le i \le K} d_{ti} $, $\check{d}_{r} = \min_{1\le i \le K}
d_{ri}$, $\hat{d}_{r} = \max_{1\le i \le K} d_{ri}$. Since increasing
the power constraints and reducing the noise variances cannot reduce
the capacity region of the GIFC we have 
\beq
\creg(H, (1/\check{d}_r^2) \mb{1}, P \check{d}_t^2 \mb{1}) \subseteq \creg(H, \mb{1}^t D_r^{-2}, P \mb{1}^tD_t^2) \subseteq \creg(H, (1/\hat{d}_r^2) \mb{1}, P \hat{d}_t^2 \mb{1}).
\label{eq:norm1}
\eeq
Furthermore, once all the noise variances are equal, they can be normalized to 1 by scaling the power constraints, leading to
\beqa
\creg(H, (1/\check{d}_r^2) \mb{1}, P \check{d}_t^2 \mb{1}) &=& \creg(H, \mb{1}, P \check{d}_r^2 \check{d}_t^2 \mb{1})\nonumber\\
\creg(H, (1/\hat{d}_r^2) \mb{1}, P \hat{d}_t^2\mb{1}) &=& \creg(H, \mb{1}, P \hat{d}_r^2 \hat{d}_t^2 \mb{1}).
\label{eq:norm2}
\eeqa
Using (\ref{eq:eq1}), (\ref{eq:norm1}) and (\ref{eq:norm2}) we can write

\beq
\frac{\max_{\mathbf{R} \in \creg(H, \mb{1}, P \check{d}_r^2 \check{d}_t^2 \mb{1})} \mathbf{1}^t \mathbf{R}}{\frac{1}{2} \log_2 P} \le \frac{\max_{\mathbf{R} \in \creg(D_tHD_R, \mb{1}, P\mb{1})} \mathbf{1}^t \mathbf{R}}{\frac{1}{2} \log_2 P} \le \frac{\max_{\mathbf{R} \in \creg(H, \mb{1}, P \hat{d}_r^2 \hat{d}_{t}^2\mb{1})} \mathbf{1}^t \mathbf{R}}{\frac{1}{2} \log_2 P}\nonumber
\eeq
which can be rewritten as
\begin{multline}
\frac{\frac{1}{2}\log_2(P
  \check{d}^2_r \check{d}^2_t) }
{\frac{1}{2} \log_2 P } 
\frac{\max_{\mathbf{R} \in \creg(H, \mb{1}, P\check{d}^2_r\check{t}^2_t
    \mb{1})} \mathbf{1}^t \mathbf{R}}{\frac{1}{2} \log_2(P
  \check{d}^2_r \check{d}^2_t) } \\ 
\le \frac{\max_{\mathbf{R} \in
    \creg(D_tHD_R, \mb{1}, P)} \mathbf{1}^t
  \mathbf{R}}{\frac{1}{2} \log_2 P} \le 
\frac{\frac{1}{2}\log_2(P
  \hat{d}^2_r \hat{d}^2_t) }
{\frac{1}{2} \log_2 P } 
\frac{\max_{\mathbf{R} \in
    \creg(H, \mb{1}, P\hat{d}^2_r\hat{d}^2_t \mb{1})} \mathbf{1}^t
  \mathbf{R}}{\frac{1}{2} \log_2 (P\hat{d}^2_r \hat{d}^2_t)}.
\nonumber 
\end{multline}
Taking $ \limsup $ of all three terms as $ P \rightarrow \infty $
implies $ DoF(H) \leq DoF(D_t H D_r) \leq DoF(H) $.
\end{proof}

\begin{proof}(of Theorem~\ref{thm:irrational})
By Lemma~\ref{lem:invariance} we can scale $ H $ (by post multiplying by an integer valued $ D_r $) so that all off-diagonal elements are integers and all diagonal elements remain irrational algebraic.  In addition, from (\ref{eq:dofdef}) we only need to consider channels where all the inputs have the same power constraint $P$ and all the noise processes have variance 1.

For any $\epsilon > 0 $, we will present a communication scheme that
achieves $\mathbf{1}^t\mathbf{R} = (K/4 -  K \epsilon)\log_2 P -
o(\log_2 P)$, implying that $ DoF(H) \geq K/2 $. Consider the scalar lattice
\[
\Lambda_{P,\epsilon} = \{x: x=P^{1/4+\epsilon} z, z\in \intgrs\}
\]
and let $\mathcal{C}_{P,\epsilon}=\Lambda_{P,\epsilon} \cap [-\sqrt{P}, \sqrt{P}]$. Note that
\beq
|\mathcal{C}_{P,\epsilon}|= 2 \left\lfloor \frac{\sqrt{P}}{P^{1/4+\epsilon}}\right\rfloor+1 \le 2 P^{1/4-\epsilon}+1.
\label{eq:cardC}
\eeq
The users communicate using codebooks of block length $n$, obtained by
uniform i.i.d. sampling $\mathcal{C}_{P,\epsilon}$. Note
that due to the truncation of the lattice to the interval $[-\sqrt{P},
  \sqrt{P}]$, the symbol power ($x_{i,t}^2$) never exceeds $P$ at any
time index, and 
hence the average codeword power does not exceed $P$. Each receiver
decodes the signal of its transmitter, treating the interfering
signals as i.i.d. noise. With this scheme, as $n\to \infty$ we can
achieve: 
\[
R_i = I(X_i;Y_i)= H(X_{i})- H(X_{i}|Y_{i})\text{ , }i=1,\ldots,K,
\]
where $X_{i} \sim \text{Uniform}(\mathcal{C}_{P,\epsilon})$,
$Y_i=\sum_{j=1}^K h_{ji} X_j+Z_i$, and $Z_i\sim \n(0,1)$,
$i=1,\ldots,K$. 

First we note that $H(X_{i})=\log_2|\mathcal{C}_{P,\epsilon}| \approx (\frac{1}{4}-\epsilon) \log_2 P+\log_2 2 $. We will show that 
\[
\limsup_{P\to \infty}H(X_{i}|Y_{i})\le 1 \text{ , } i=1,\ldots,K,
\]
and as a result, $R_i = (\frac{1}{4}-\epsilon) \log_2 P -o(\log_2 P)$
can be achieved. It would then follow that $DoF(H)\ge K/2$, and, from
the upper bound of~\cite{HN05}, that $ DoF(H) = K/2 $.

We will use the following lemma to upper bound $H(X_i|Y_i)$ for $i=1,\ldots,K$. 

\begin{lemma}
\label{lemma:latticesum}
Let $\Sigma_{P,\epsilon}=\{\alpha x+s: x\in
\mc{C}_{P,\epsilon},s \in \Lambda_{P,\epsilon} \}$, with $\alpha$
being any real, irrational, and algebraic number. For any $y \in
\Sigma_{P,\epsilon}$ there exists a unique pair $(x,s) \in
\mathcal{C}_{P,\epsilon}\times \Lambda_{P,\epsilon}$ such that
$y=\alpha x+s$. In addition, if $y_1,y_2 \in \Sigma_{P,\epsilon}$,
$y_1\ne y_2$, then $|y_1-y_2| > P^\epsilon$ for any given $\epsilon>0$
and large enough $P$. 
\end{lemma}
\begin{proof}
Let $y=\alpha x+s$ with $x\in \mathcal{C}_{P,\epsilon}$, $s \in
\Lambda_{P,\epsilon}$.  To get a contradiction, assume that there
exists $(\tilde{x}, \tilde{s}) \in \mathcal{C}_{P,\epsilon} \times
\Lambda_{P,\epsilon}$ with $(\tilde{x}, \tilde{s})\ne (x,s)$ such that
$\alpha \tilde{x}+\tilde{s}=y$. Without loss of generality we can
assume $\tilde{x} \ge x$. Since $\alpha \ne 0$ we have $\tilde{s} \ne
s$ and $\tilde{x} > x$. In addition, since by assumption $\alpha x+ s
= \alpha \tilde{x}+\tilde{s}$, we have  
\[
\alpha = \frac{s - \tilde{s}}{\tilde{x}-x} = \frac{(z_s - z_{\tilde{s}})P^{1/4+\epsilon}}{(z_{\tilde{x}}-z_x)P^{1/4+\epsilon}} = \frac{z_s - z_{\tilde{s}}}{z_{\tilde{x}}-z_x} \in \rationals
\]
where $z_x, z_{\tilde{x}}, z_s, z_{\tilde{s}} \in \intgrs$, which contradicts the assumption of irrational $\alpha$.

To prove the second part of the lemma, let $\hat{y}= \alpha \hat{x}+\hat{s}$, where $\hat{x}\in\mc{C}_{P,\epsilon}$, $\hat{s} \in \Lambda_{P,\epsilon}$, and $\hat{y} \in \Sigma_{P,\epsilon}$, with $\hat{y} \ne y$.  If $\hat{s} = s$ then 
\[
|\hat{y}-y|= \alpha |\hat{x}-x|=\alpha |z_{\hat{x}}-z_x|
 P^{1/4+\epsilon} > P^{\epsilon},
\]
where $z_x, z_{\hat{x}} \in \intgrs$, as long as $P$ is sufficiently large. Similarly, if $\hat{x} = x$ and $P$ is large enough we have
\[
|\hat{y}-y|= |\hat{s}-s|=|z_{\hat{s}}-z_s| P^{1/4+\epsilon} > P^{\epsilon},
\]
where $z_s, z_{\hat{s}} \in \intgrs$. So it remains to consider the case $\hat{x} \ne x$ and $\hat{s} \ne s$. Without loss of generality we can assume $\hat{x} > x$. To get a contradiction, we assume that $|\hat{y}-y| \le P^\epsilon$, and write:
\beqa
|\hat{y}-y| &\le& P^\epsilon \nonumber\\
|\alpha \hat{x} + \hat{s} - \alpha x - s| &\le& P^\epsilon \nonumber\\
|\alpha z_{\hat{x}}+z_{\hat{s}} - \alpha z_x - z_s| &\le& \frac{P^\epsilon}{P^{1/4+\epsilon}} \nonumber\\
\left|\alpha - \frac{z_s-z_{\hat{s}}}{z_{\hat{x}}-z_x}\right| &\le&
\frac{P^{-1/4}}{z_{\hat{x}}-z_x},
\label{eq:error}
\eeqa
where $z_x, z_{\tilde{x}}, z_s, z_{\tilde{s}} \in \intgrs$.

On the other hand there are bounds on how well an irrational algebraic number can be approximated with a rational number. The most refined of those bounds, due to Roth, 1955, states that for any irrational algebraic $\alpha$, and any $\gamma >0$, there exists $\delta > 0$ such that
\beq
\left|\alpha -\frac{p}{q}\right| > \frac{\delta}{q^{2+\gamma}}
\label{eq:roth}
\eeq
for all $p,q \in \intgrs$, $q > 0$ \cite{Sta84}.

Combining (\ref{eq:error}) and (\ref{eq:roth}) we have
\[
\frac{\delta}{(z_{\hat{x}}-z_x)^{2+\gamma}} < \left|\alpha - \frac{z_s-z_{\hat{s}}}{z_{\hat{x}}-z_x}\right| \le \frac{P^{-1/4}}{z_{\hat{x}}-z_x}
\]
so that
\beq
0 < \delta <  P^{-1/4} (z_{\hat{x}}-z_x)^{1+\gamma} \stackrel{(a)}{\le} P^{-1/4} \left(2 P^{1/4-\epsilon}+1\right)^{1+\gamma} = 2^{1+\gamma}P^{\gamma/4-\epsilon(1+\gamma)}+o(1)
\label{eq:delta}
\eeq
where we used (\ref{eq:cardC}) in step (a). But the right hand side of (\ref{eq:delta}) goes to 0 as $P\to \infty$ whenever $\epsilon \ge 1/4$ or $\gamma < \epsilon/(1/4-\epsilon)$. Since we can choose any $\gamma > 0$, we can obtain a contradiction in (\ref{eq:delta}) for any $\epsilon>0$, for large enough $P$.
\end{proof}

We will use Lemma \ref{lemma:latticesum} to build an estimator that can identify $X_i$ in $Y_i$ with high probability.

Let $S_i\dfn \sum_{j\ne i} h_{ji} X_j$, and note that since $h_{ji}\in
\intgrs$ for $j \ne i$ we have that $S_i \in \Lambda_{P,\epsilon}$. In
addition, let $\Sigma_{P,\epsilon,i} = \{h_{ii}x + y:
x \in \mc{C}_{P,\epsilon},y \in \Lambda_{P,\epsilon}\}$, and let
$v_i:\mathcal{C}_{P,\epsilon}\times \Lambda_{P,\epsilon} \to
\Sigma_{P,\epsilon,i}$ be defined as $v_i(x,s) = h_{ii}x + s$. Using
Lemma \ref{lemma:latticesum} and the fact that $h_{ii}$ is real,
algebraic and irrational we have that $v_i$ is invertible, i.e. there
exists $v_i^{-1}:\Sigma_{P,\epsilon,i} \to \mathcal{C}_{P,\epsilon}
\times \Lambda_{P,\epsilon}$ such that $v_i^{-1}(v_i(x,s))=(x,s)$ for
any $(x,s) \in \mathcal{C}_{P,\epsilon} \times \Lambda_{P,\epsilon}$.  

Let
$u:\mathbb{R}^2\to \mathbb{R}$ be defined as $u(x,s)=x$, and  
let $\hat{X}_{i}=u(v_i^{-1}(\arg
  \min_{x\in \Sigma_{P,\epsilon,i}} |x-Y_{i}|))$. We have $\hat{X}_{i}
  \ne X_{i}$ whenever $Y_{i}$ is closer to some other point in
  $\Sigma_{P,\epsilon,i}$ than it is to $h_{ii} X_{i}+ S_i$. From
  Lemma \ref{lemma:latticesum}, this can only occur if $|Z_{i}| \ge
  P^{\epsilon}/2$ for large enough $P$. It follows that 
\[
Pr(\hat{X}_{i} \ne X_{i}) \le Pr \left(|Z_{i}| \ge
\frac{P^{\epsilon}}{2}\right) = 2
Q_{\mathcal{N}(0,1)}\left(\frac{P^{\epsilon}}{2} \right) \le 2
\exp\left(-\frac{P^{2\epsilon}}{8}\right), 
\]
where $ Q_{\mathcal{N}(0,1)}(x) $ is the probability that a Gaussian
random variable with zero mean and variance one exceeds $ x $.
Using the data processing and Fano's inequalities we obtain
\beqa
H(X_{i}|Y_{i}) &\le& H(X_{i}|\hat{X}_{i}) \nonumber\\
&\le& 1+ Pr(\hat{X}_{i}\ne X_{i}) \log(|\mathcal{C}_{P,\epsilon}|) \nonumber\\
&\le& 1+2 \exp\left(-\frac{P^{2\epsilon}}{8}\right) \left[\left(\frac{1}{4}-\epsilon\right) \log_2 P +\log_2 2+o(1)\right]
\eeqa
which goes to $1$ as $P\to \infty$.
\end{proof}

\section{Degrees-of-freedom for rational $ H $}
\label{sec:rational}
In this section, we give the proof of Theorem~\ref{thm:rational}, 
establishing that the degrees-of-freedom of any fully connected, real, scalar
GIFCs is bounded strictly below $ K/2 $, for $ K > 2$.  As
in the previous section, we begin with a sketch of the proof.

Most
of the work in the proof is to establish the theorem for $ K = 3 $
users.  The theorem for $ K > 3 $ will then follow from an extension of
the averaging argument of~\cite{HN05}, used therein to obtain the $ K/2
$ degrees-of-freedom upper bound from a bound of $ 1 $ on the
degrees-of-freedom for $ K=2 $.  In this case, the $ K = 3 $ bound is
averaged over all three-tuples of users (transmitters and corresponding
receivers), as opposed to pairs of users in~\cite{HN05}.  

Given a  $ K = 3 $ user GIFC with fully connected, rational $ H $, using
Lemma~\ref{lem:invariance} (invariance property) and eliminating cross
links, we can upper bound $ DoF(H) $ by $ DoF(\tilde{H}) $ where
\[
\tilde{H} = \big[\tilde{h}_{ij}\big]=\left[
\begin{array}{ccc} 
1 & 0 & 0 \\ 
1 & p & 0 \\ 
1 & q & 1 
\end{array} \right],
\]
where $ p $ and $ q $ are integers (see
Figure~\ref{fig:3x3channel} in Subsection~\ref{sec:mainlemma}).  
The main step in the proof of the overall theorem
is establishing that $ DoF(\tilde{H}) < 3/2 $, which is 
formally carried out in Lemma~\ref{lem:mainlemma} below, and proceeds
as follows.  First, it is shown (Lemma~\ref{lemma:deterministic})
that a deterministic channel obtained
by eliminating all noise sources
and restricting the power-constrained codewords to have integer
valued symbols results in at most a power-constraint-independent loss in the
achievable sum rate.  Therefore, the degrees-of-freedom (according to
the obvious generalization) of this deterministic interference channel
(IFC) is no
smaller than $ DoF(\tilde{H}) $.  Next, it is shown using a 
Fano's inequality based argument that if the degrees-of-freedom of
the deterministic IFC 
is at least $ 3/2 $ there would exist finite sets of $ n
$-dimensional integer 
valued vectors $ \mc{X}_2 $ and $  \mc{X}_3
$ such that the corresponding independent random variables 
$\bx^n_2 $ and $ \bx^n_3 $, uniformly distributed on these sets, induce
discrete entropies satisfying $ H(\bx^n_2) \approx n(1/4)\log_2 P $, $
H(\bx^n_3) \approx n(1/4)\log_2 P $, $ H(\bx^n_2 + \bx^n_3) \approx
n((1/4)+\epsilon)\log_2 P $, 
and  $ H(p\cdot \bx^n_2 + q\cdot \bx^n_3) \approx n((1/2) - \epsilon)\log_2
P $, for the integers $ p $ and $ q $ defining the channel and $
\epsilon $ arbitrarily small.  These entropy relations suggest that
the cardinality of the support of $p\cdot \bx^n_2 + q\cdot \bx^n_3 $ is
much larger than that of $ \bx^n_2 + \bx^n_3 $.  However, tools from additive
combinatorics (through Lemma~\ref{lem:setsum})
can be used to show that this is impossible for integer
valued $ p $ and $ q $, leading to a 
 contradiction, and thereby implying
that the deterministic channel must have degrees-of-freedom strictly
smaller than $ 3/2 $.   Unfortunately, the link between the
entropy and the cardinality of the support of a sum of independent, uniformly
distributed random
variables is sufficiently weak that a somewhat more involved argument
(incorporating Lemma~\ref{lemma:exG} and Theorem~\ref{thm:balog})
is ultimately required to reach the above conclusions.  The overall
intuition behind the proof, however, is as outlined.

The rest of the section is organized as follows.  
Subsections~\ref{sec:infotheory} and~\ref{sec:addcomb} respectively
collect supporting results of an information theoretic nature and results from
additive combinatorics.  The proof of the main lemma on the
$ K = 3 $ user IFC is presented in Subsection~\ref{sec:mainlemma}.
Finally, the extension to $ K > 3 $ is presented in
Subsection~\ref{sec:genK}.  Throughout, the capacity region of a $ K
$-user GIFC will be taken as in Definition~\ref{def:capreg}.

\subsection{Supporting information theoretic results}
\label{sec:infotheory}
\begin{lemma}
The capacity region
of a $K$-user memoryless IFC,\footnote{Here we are
  considering more general IFCs 
  than the Gaussian case.  Definition~\ref{def:capreg} still applies,
  but with the appropriate conditional probability distribution
  of channel outputs given channel inputs.}
 where the codebook of user $i$ is subject to an average power constraint $P_i$, is given by the limiting expression:
\beq
{\cal C}_{IFC}=\bigcup_{n=1}^\infty \bigcup_{\stackrel{P_{\bx_1^n
      \ldots \bx_K^n}=P_{\bx_1^n}\cdot\cdot\cdot
    P_{\bx_K^n}}{Pr(\|\bx_i^n\|^2_2\le n P_i)=1 \text{ , } i=1,\ldots,
    K}} \left\{\mb{R}\in \Rea_+^K: R_i \le \frac{1}{n} I(\bx_i^n;
\by_i^n), i=1,\ldots, K \right\}  
\label{eq:limcap}
\eeq
\end{lemma}
\begin{proof}
The lemma can be proved by extending the argument of \cite{Ahl71} to $K$-user memoryless IFCs with possibly continuous alphabets and average power constraints on the inputs. The details are omitted. 
\end{proof}

\begin{lemma}
\label{lemma:deterministic}
Given a gain matrix $ H $ and power constraints $ \mathbf{P} =
(P_1,\ldots,P_K) $, let $ {\cal C}_D(H,\mathbf{P}) $ denote the
capacity region of the deterministic IFC defined by
\[
\bar{y}_i(t) = \sum_{j=1}^K h_{ji} \bar{x}_j(t) \text{ , } i=1,\ldots,K
\]
where the inputs are constrained to be integers (i.e. $\bar{x}_i(t) \in \mathbb{Z}$, $i=1,\ldots,K$, $t=1,2\ldots$) and satisfy an average power constraint $\frac{1}{n}\sum_{t=1}^n \bar{x}_i(t)^2 \le P_i$ for all $i$.

Then, $\mathbf{R} \in \creg(H,\mathbf{1},\mathbf{P})\Rightarrow
(\mathbf{R}-\boldsymbol{\Delta}) \in \creg_D(H,\mathbf{P})$ with
$\boldsymbol{\Delta}=(\delta_1,\ldots,\delta_K)$, $\delta_i =
\frac{1}{2} \log_2(1+2 \sum_{j=1}^K h_{ji}^2)$, $i=1,\ldots,K$, where
$ \creg(H,\mathbf{1},\mathbf{P}) $ is the capacity region of the
corresponding GIFC (see Definition~\ref{def:capreg}).
\end{lemma}
\begin{proof}
Let $f:\mathbb{R} \to \mathbb{R}$ be defined as $f(x)\dfn \lfloor x
\rfloor \cdot 1(x>0) + \lceil x \rceil \cdot 1(x<0)$, and let
$g:\mathbb{R} \to \mathbb{R}$ be defined as $g(x) = x - f(x)$. In
addition, let $x_{i1}=f(x_i)$, $x_{i2}=g(x_i)$, $y_{i1}=\sum_{j=1}^K
h_{ji} x_{j1}+z_{i1}$, and $y_{i2}=\sum_{j=1^K} h_{ji} x_{j2}+z_{i2}$,
where $z_{i1},z_{i2}\sim \n(0,1/2)$ are independent. Then the outputs
of the $K$-user Gaussian IFC can be written as $y_i = y_{i1}+y_{i2}$,
for $i=1,\ldots, K$ (see Figure \ref{fig:KxKchannel}). 

\begin{figure}[t]
\centerline{\includegraphics[width=3in]{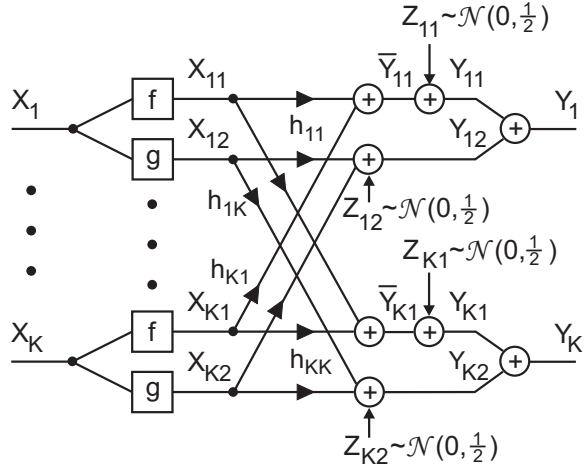}}
\caption{\footnotesize A decomposition of a $K$-user Gaussian IFC. \label{fig:KxKchannel} }
\end{figure}

If $\mathbf{R}=(R_1,\ldots,R_K) \in \creg(H,\mathbf{1},\mathbf{P})$,
then for any $\eta > 0$ there exists a family of codebooks
$\{C_{1,n},\ldots,C_{K,n}\}_n$ satisfying the average power
constraints, and decoding functions $\{ g_{1,n},\ldots,g_{K,n}\}_{n}$
with average decoding error probability going to 0 as $n\to \infty$,
such that $\lim_{n\to \infty} \frac{1}{n}\log_2|C_{i,n}| \ge R_i -
\eta$.  
For block-length $n$ we have:
\begin{align}
n(R_i-\eta-&\epsilon_n) \nonumber\\
\le& I(\bx_i^n;\by_i^n) \label{eq:fano}\\
\le& I(\bx_{i1}^n, \bx_{i2}^n;\by_{i1}^n,\by_{i2}^n) \label{eq:datapr}\\
=& h(\by_{i1}^n, \by_{i2}^n) - h(\by_{i1}^n, \by_{i2}^n| \bx_{i1}^n, \bx_{i2}^n)\nonumber\\
=& h(\by_{i1}^n)+h(\by_{i2}^n|\by_{i1}^n) - h\Bigg(\sum_{\stackrel{j=1}{j\ne i}}^K h_{ji}\bx_{j1}^n+\bz_{i1}^n, \sum_{\stackrel{j=1}{j\ne i}}^K h_{ji}\bx_{j2}^n+\bz_{i2}^n\Bigg) \nonumber\\
\le& h(\by_{i1}^n)+h(\by_{i2}^n) - h\Bigg(\sum_{\stackrel{j=1}{j\ne i}}^K h_{ji}\bx_{j1}^n +\bz_{i1}^n\Bigg)- h\Bigg(\sum_{\stackrel{j=1}{j\ne i}}^K h_{ji}\bx_{j2}^n+\bz_{i2}^n \bigg| \sum_{\stackrel{j=1}{j\ne i}}^K h_{ji}\bx_{j1}^n+\bz_{i1}^n\Bigg) \label{eq:cond1}\\
\le& h(\by_{i1}^n)+h(\by_{i2}^n)- h\Bigg(\sum_{\stackrel{j=1}{j\ne i}}^K h_{ji}\bx_{j1}^n +\bz_{i1}^n\Bigg)\nonumber\\
&- h\Bigg(\sum_{\stackrel{j=1}{j\ne i}}^K h_{ji}\bx_{j2}^n+\bz_{i2}^n \bigg| \sum_{\stackrel{j=1}{j\ne i}}^K h_{ji}\bx_{j1}^n+\bz_{i1}^n,\bx_{12}^n,\ldots,\bx_{K2}^n\Bigg) \nonumber\\
\label{eq:cond2}\\
=& h(\by_{i1}^n)+h\Bigg(\sum_{j=1}^K h_{ji}\bx_{j2}^n+\bz_{i2}^n\Bigg)- h\Bigg(\sum_{\stackrel{j=1}{j\ne i}}^K h_{ji}\bx_{j1}^n+\bz_{i1}^n\Bigg)- h(\bz_{i2}^n) \nonumber\\
\le& I(\bx_{i1}^n; \by_{i1}^n)+ \frac{n}{2} \log_2\bigg[2 \pi e \bigg(\sum_{j=1}^K h_{ji}^2+\frac{1}{2}\bigg) \bigg] - \frac{n}{2} \log_2\bigg(2 \pi e \frac{1}{2}\bigg) \label{eq:gbound}\\
\le& I(\bx_{i1}^n; \bar{\by}_{i1}^n) + \frac{n}{2} \log_2\bigg(1+2\sum_{j=1}^K h_{ji}^2\bigg)
\label{eq:datapr2}
\end{align}
where $\epsilon_n\to 0$ as $n\to \infty$. We used Fano's inequality in
(\ref{eq:fano}), the data processing inequality in (\ref{eq:datapr}),
the fact that conditioning reduces entropy in (\ref{eq:cond1}) and
(\ref{eq:cond2}), the Gaussian bound for differential entropies in
(\ref{eq:gbound}), noting that $|x_{j2}(t)| \le 1$, $t=1,\ldots,n$,
and used the data processing inequality in (\ref{eq:datapr2}), where
we defined $\bar{\by}_{i1}^n\dfn\sum_{j=1}^n h_{ji} \bx_{j1}^n$. 

Since $\|\bx_{i1}^n\|_2^2 \le \|\bx_{i}^n\|_2^2 \le n P_i$ when the
channel is used with codebooks satisfying the power constraints, it
follows that the codebooks
$\{C_{1,n},\ldots,C_{K,n}\}_n$ induce
distributions on $\bx_{i1}^n$, $i=1,\ldots, K$, that satisfy the power
constraints $\{P_i\}_{i=1}^K$. These distributions can be used in
(\ref{eq:limcap}) to conclude that $(R_1-\eta-\delta_1,\ldots,
R_K-\eta-\delta_K,\ldots)$ is an achievable rate vector in
$\creg_D(H,\mathbf{P})$. Since $\eta>0$ is arbitrary, the result
follows.   
\end{proof}
\begin{lemma}
\label{lemma:pruned_codebooks}
Let $\{(C_{i,n},g_{i,n}, P_{e,i,n})\}_{i=1}^K$ denote a block length
$n$ coding scheme for a $K$-user interference channel satisfying
average (per codeword) power constraints $\{P_i\}_{i=1}^K$ with rates
$\{R_i\}_{i=1}^K$ and average error probabilities
$\{P_{e,i,n}\}_{i=1}^K$. If
$\{\tilde{C}_{1,n},\ldots,\tilde{C}_{K,n}\}$ is any set of codebooks
with $\tilde{C}_{i,n} \subseteq C_{i,n}$, $|\tilde{C}_{i,n}| \ge
|C_{i,n}|/\alpha_i$ and $\alpha_i \ge 1$, for all $i=1,\ldots, K$,
then $\{(\tilde{C}_{i,n},g_{i,n},\tilde{P}_{e,i,n} )\}_{i=1}^K$ is a
coding scheme with rates no smaller than $\{R_i-\frac{1}{n}\log_2
\alpha_i\}_{i=1}^K$ and average error probabilities
$\{\tilde{P}_{e,i,n}\}_{i=1}^K$ satisfying $\tilde{P}_{e,i,n}
\le(\prod_{j=1}^K \alpha_j) P_{e,i,n}$, and also satisfying the power
constraints $\{P_i\}_{i=1}^K$. 
\end{lemma}
\begin{proof}
Since $C_{i,n}$ has rate $R_i$, we have that for $i=1,\ldots, K$ the rate $\tilde{R}_i$ of $\tilde{C}_{i,n}$ satisfies:
\beq
\tilde{R}_i=\frac{1}{n}\log_2 |\tilde{C}_{i,n}| \ge \frac{1}{n}\log_2\left( \frac{|C_{i,n}|}{\alpha_i}\right) = R_i - \frac{1}{n} \log_2 \alpha_i.
\eeq

During communication with the codebooks $\{C_{1,n},\ldots,C_{K,n}\}$ the transmitted messages (and hence the codewords) are chosen uniformly and independently, and as a result we have:
\beqa
P_{e,i,n} &=& \frac{1}{\prod_{j=1}^K |C_{j,n}|} \sum_{\mb{c}_1\in C_{1,n}} \cdots \sum_{\mb{c}_K\in C_{K,n} } P_{e,i,n}(\mb{c}_1,\ldots,\mb{c}_K) \nonumber\\
&\ge& \frac{1}{\prod_{j=1}^K |C_{j,n}|} \sum_{\mb{c}_1\in \tilde{C}_{1,n}} \cdots \sum_{\mb{c}_K\in \tilde{C}_{K,n} } P_{e,i,n}(\mb{c}_1,\ldots,\mb{c}_K) \nonumber\\
&=& \frac{1}{\prod_{j=1}^K \alpha_j}\frac{\prod_{j=1}^K \alpha_j} {\prod_{j=1}^K |C_{j,n}|} \sum_{\mb{c}_1\in \tilde{C}_{1,n}} \cdots \sum_{\mb{c}_K\in \tilde{C}_{K,n} } P_{e,i,n}(\mb{c}_1,\ldots,\mb{c}_K) \nonumber\\
&\ge& \frac{1}{\prod_{j=1}^K \alpha_j}\frac{1} {\prod_{j=1}^K |\tilde{C}_{j,n}|} \sum_{\mb{c}_1\in \tilde{C}_{1,n}} \cdots \sum_{\mb{c}_K\in \tilde{C}_{K,n} } P_{e,i,n}(\mb{c}_1,\ldots,\mb{c}_K) \nonumber\\
&=& \frac{1}{\prod_{j=1}^K \alpha_j} \tilde{P}_{e,i,n}
\eeqa
where we denoted by $P_{e,i,n}(\mb{c}_1,\ldots,\mb{c}_K)$ the probability of decoding error when the codewords $\mb{c}_1,\ldots,\mb{c}_K$ are transmitted. 

Finally, since every codeword of $C_{i,n}$ satisfies the power constraint $P_i$, the codewords of $\tilde{C}_{i,n}$ satisfy the power constraint $P_i$.
\end{proof}

\begin{lemma}
\label{lemma:repeated_codewords}
Any achievable rate vector in a $K$-user IFC can be achieved by codebooks with no repeated codewords, i.e. for every $n=1,2,\ldots$ and $k=1,\ldots,K$, the codebook $C_{k,n}$ is such that $\mb{c}_i,\mb{c}_j \in C_{k,n} \Rightarrow \mb{c}_i \ne \mb{c}_j$. 
\end{lemma}
\begin{proof}
Since $(R_1,\ldots, R_K)$ is achievable, for any $\eta > 0$ there exists a family of codebooks $\{C_{1,n},\ldots,C_{K,n}\}_n$ with $\frac{1}{n} \log_2|C_{i,n}| \ge R_i-\eta$ satisfying the average power constraints, and a family of decoding functions achieving average error probabilities $P_{e,i,n}$ going to 0 as $n\to \infty$ for every $i=1,\ldots,K$.

Consider the single user channel between transmitter $i$ and receiver
$i$ obtained from the interference channel by removing all the
interfering signals at receiver $i$. Since the interference cannot
help receiver $i$ decode the message of its own transmitter, it
follows that $R_i$ can be achieved in the single user channel with the
family of codebooks $\{ C_{i,n} \}_{i,n}$, for some decoding functions
$\{g'_{i,n}\}_{i,n}$ with average error probabilities no larger than
$P_{e,i,n}$, $i=1,\ldots,K$, $n=1,2,\ldots$. Let $\tilde{C}_{i,n}
\subset C_{i,n}$ be obtained by removing the worst (i.e. leading to
the largest error probability in the single user channel) half of the
codewords in $C_{i,n}$. It is easy to see that $\tilde{C}_{i,n}$ and
$g'_{i,n}$ achieve a maximal error probability in the single user
channel no larger than $2 P_{e,i,n}$, and in particular,
$\tilde{C}_{i,n}$ has no repeated codewords for $n$ large enough. The
result follows by using Lemma \ref{lemma:pruned_codebooks} with
$\{\tilde{C}_{i,n}\}_{i,n}$ as defined here, and $\alpha_i=2$,
$i=1,\ldots,K$, noting that as $n\to \infty$, $\frac{1}{n} \log_2
\alpha_i \to 0$ and $(\prod_{j=1}^K \alpha_j) P_{e,i,n} \to 0$ for
$i=1,\ldots,K$. 
\end{proof}

\subsection{Supporting results from additive combinatorics}
\label{sec:addcomb}
Given an abelian group
$ G $
 and two sets $ A, B \subseteq G $ let $ A + B $ denote the set of
 sums obtainable by adding one element from $ A $ to one element from
 $ B $.\footnote{We shall apply these results
  to the group of vectors of integers with component-wise   addition.} 
  Formally, $ A + B = \{a+b:a\in A,b \in B\} $.  The set of differences $ A - B $ can be defined analogously.
For any integer $ p $ and set $ A \subseteq G $, we denote by $ p\cdot
A $ the set consisting of all $ p $-multiples of elements of $ A $ or
$ p\cdot A = \{pa:a \in A\} $. For a non-negative integer $ p $ and $ A
\subseteq G $ we denote by $ p \star A $ the set of $p$-fold sums of $
A $ or $ p\star A = 
\{a_1+a_2+\ldots+a_p: a_i \in A \text{ for } i = 1,\ldots,p\}$. 
We shall also require the concept of a partial sum set.  Given $ A,B
  \subseteq G $, let $ F \subseteq A \times B $.  The partial sum set
  of $ A $ and $ B $ with respect to $ F $, denoted as $ A
  \stackrel{F}{+} B $, is defined as $ A \stackrel{F}{+} B =
  \{a+b:(a,b) \in F\} $.  If $ F = A\times B $, then $ A \stackrel{F}{+}
  B = A+B $.

We shall later need the following result on sum sets.
\begin{lemma}
\label{lem:setsum}
Let $ p,q \in \intgrs $ and $ K \in \Rea $ with $ K \geq 1$.
If $ |A+B| \leq K |A|^{1/2}|B|^{1/2} $ then $ |p\cdot A + q\cdot B|
\leq K^{d(p,q)}|A|^{1/2}|B|^{1/2} $ for $ d(p,q) = 2\max\{|p|,|q|\}+5 $.
\end{lemma}
Our proof, stated below, follows fairly standard arguments from 
additive combinatorics (see Chapters~2 and~6 of~\cite{TaoVu06}) and is based
on the following key results concerning non-empty subsets $ A, B \subseteq G$.  Proofs of these results can be found in~\cite{TaoVu06}.
\begin{lemma}[Rusza's covering lemma]
\label{lem:covering}
There exists a subset $ X \subseteq B $
such that $ |X| \leq |A+B|/|A| $ and $ B \subseteq A-A+X $.
\end{lemma}
\begin{lemma}[Pl\"{u}nnecke-Rusza inequality]
\label{lem:plun}
For positive integers $ p,q $ and any real valued $ \tilde{K} \geq 1 $,
if $ |A+B| \leq \tilde{K} |A| $ then $ |p\star B-q\star B| \leq
\tilde{K}^{p+q} |A| $. 
\end{lemma}
\begin{proof} (of Lemma~\ref{lem:setsum}.)
By Lemma~\ref{lem:covering}, there exists a set $ X \subseteq q\cdot B $ with
$ |X| \leq |A+q\cdot B|/|A| $ satisfying $ q\cdot B \subseteq A - A + X $.  This, in turn, implies that $ p\cdot A + q\cdot B \subseteq p\cdot A + A - A + X $. Therefore,
\begin{align}
|p\cdot A + q\cdot B| &\leq |p\cdot A + A - A||X| \nonumber \\
&\leq |p\cdot A + A - A|\frac{|A+q\cdot B|}{|A|}\label{eq:covbnd1},
\end{align}
where we have used the trivial inequality $ |S+T| \leq |S||T|$ in the
first step.    
Again, by Lemma~\ref{lem:covering}, there exists a set $ Y \subseteq A $ with 
$ |Y| \leq |A+B|/|B| $ satisfying $ A \subseteq B-B+Y $, which implies that
$ A + q\cdot B \subseteq q\cdot B + B - B + Y $.  Proceeding as above,
we then have 
\begin{align}
|A + q\cdot B| &\leq |q\cdot B + B - B||Y| \nonumber \\
&\leq |q\cdot B + B - B|\frac{|A+B|}{|B|}\label{eq:covbnd2}.
\end{align}
Combining (\ref{eq:covbnd1}) and~(\ref{eq:covbnd2}) gives
\begin{align}
|p\cdot A + q\cdot B| &\leq |p\cdot A + A - A||q\cdot B + B -
B|\frac{|A+B|}{|A||B|}\nonumber \\ 
&\leq |p\cdot A + A - A||q\cdot B + B - B|K|A|^{-1/2}|B|^{-1/2}
\label{eq:combined1} \\
&\leq |p\star A + A - A||q\star B + B - B|K|A|^{-1/2}|B|^{-1/2}
\label{eq:combined2}
\end{align}
where (\ref{eq:combined1}) follows from the assumption of the
lemma and~(\ref{eq:combined2}) follows from the trivial inclusions $
p\cdot A + A - A \subseteq p\star A + A - A $ and $ q\cdot B + B - B
\subseteq q\star B+B-B $.  Rewriting the assumption of the lemma as
$ |A+B| \leq K|A|^{-1/2}|B|^{1/2}|A| $, we can apply
Lemma~\ref{lem:plun} with $ \tilde{K} = K|A|^{-1/2}|B|^{1/2} $ to
conclude
that
\begin{equation}
|q\star B + B - B| \leq K^{|q|+2}|A|^{-(|q|+2)/2}|B|^{(|q|+2)/2}|A|.
\label{eq:qBbnd}
\end{equation}
Similarly, rewriting the assumption of the lemma as
$ |A+B| \leq K|A|^{1/2}|B|^{-1/2}|B| $, we can apply
Lemma~\ref{lem:plun}, with the roles of $ A $ and $ B $ switched, to
obtain
\begin{equation}
|p\star A + A - A| \leq K^{|p|+2}|A|^{(|p|+2)/2}|B|^{-(|p|+2)/2}|B|.
\label{eq:pAbnd}
\end{equation}
Combining~(\ref{eq:combined2}) with~(\ref{eq:qBbnd})
and~(\ref{eq:pAbnd}) gives
\begin{align}
|p\cdot A + q\cdot B| &\leq
K^{|p|+2}|A|^{(|p|+2)/2}|B|^{-(|p|+2)/2}
K^{|q|+2}|A|^{-(|q|+2)/2}|B|^{(|q|+2)/2}K|A|^{1/2}|B|^{1/2} \nonumber \\
&\leq K^{2\max\{|p|,|q|\}+5}|A|^{1/2}|B|^{1/2} \label{eq:laststep},
\end{align}
where~(\ref{eq:laststep}) follows from
$ K|A|^{-1/2}|B|^{1/2} \geq 1 $ and $
K|A|^{1/2}|B|^{-1/2} \geq 1 $, which in turn follow
from the lemma's assumption $
|A+B| \leq K|A|^{-1/2}|B|^{1/2}|A| = K|A|^{1/2}|B|^{-1/2}|B| $
together with the obvious relations $
|A+B|\geq |A| $ and $ |A+B|\geq |B|$.  The lemma is thus established
with $ d(p,q) = 2\max\{|p|,|q|\} + 5 $.
\end{proof}

\begin{remark}
A considerably smaller $ d(p,q) $ for larger 
$ p $ and $ q $ can be obtained by applying the bounds of~\cite{Buk08}
to the factors $  
|p\cdot A + A - A| $ and $ |q\cdot B + B - B| $ 
appearing in~(\ref{eq:combined1}).  These bounds are obtained through
a more sophisticated application of Lemma~\ref{lem:plun}, that takes
greater advantage of the structure of sets like $ p_1 \cdot A + ... + p_m
\cdot A $.  The resulting $ d(p,q) $ grows logarithmically in
$ |p| $ and $ |q|$.  It is likely that a direct application of the
technique of~\cite{Buk08} to $ p\cdot A + q\cdot B $ would even further
improve $ d(p,q) $ for large $ p $ and $ q $.
\end{remark}

We shall also make use of the following lemma relating the entropy of
a sum of uniformly distributed independent random variables to a
partial sum set involving their supports.
\begin{lemma}
\label{lemma:exG}
Let $X$ and $Y$ be independent uniform random variables with support
sets $A \subseteq G$ and $B \subseteq G$ for some abelian group $ G $,
with $|A| \ge |B|$, such that 
\[
H(X+Y) \le (1+\epsilon) \log_2 |A|
\]
for some $\epsilon > 0$. Then, for any given $c > 1$ there exists a set $F \subseteq  A\times B$ such that
\[
|F| \ge |A| |B| \frac{c-1}{c} \text{ and } |A \stackrel{F}{+} B| \le \left(|A|^{1/2+c \epsilon} |B|^{-1/2} \right) |A|^{1/2} |B|^{1/2}.
\]
\end{lemma}
\begin{proof}
Define $ T(s) = \{(a,b) \in A \times B : a + b = s\} $.
Define $S \dfn \left\{s: |T(s)|\ge |B||A|^{-c\epsilon}
\right\}$ and $F \dfn \left\{(a,b) \in A\times B: a+b \in S
\right\}$. From these definitions we have $|A \stackrel{F}{+} B| =
|S|$. In addition, 
\beqa
|A||B| &\geq& \sum_{s \in S}|T(s)| \nonumber \\
       &\geq& |S||B||A|^{-c\epsilon}, \nonumber
\eeqa
where the last step follows from the definition of $ S $.
As a result, $|S| \le
|A|^{1+c \epsilon}$, giving the required upper bound for $|A
\stackrel{F}{+} B| = |S|$. 

To get a lower bound on $|F|$ we start by rewriting it as follows:
\beqa
|F| &=& \sum_{(x,y)\in A\times B} 1((x,y)\in F) \nonumber\\
&=& |A||B| Pr((X,Y) \in F) \nonumber\\
&=& |A||B| Pr(X+Y \in S) \nonumber\\
&=& |A||B|\left[1-Pr\left(X+Y \in S^c\right)\right]\nonumber\\
&=& |A||B|\left[1-Pr\left( |T(X+Y)| < |B||A|^{-c\epsilon}\right)\right].
\label{eq:F}
\eeqa
The probability term can be upper bounded using Markov's
inequality, noting that $ |B| \geq |T(s)| $ for all $ s $.
\beqa
Pr\left( |T(X+Y)| < |B||A|^{-c\epsilon} \right) &=& Pr\left( \log_2
\frac{|B|}{|T(X+Y)| }> c \epsilon \log_2 |A|\right) \nonumber
\\ 
&\le& \frac{E\left[\log_2 \frac{|B|}{|T(X+Y)|}\right]}{c
  \epsilon \log_2 |A|}. 
\label{eq:markov}
\eeqa
To bound the expectation in (\ref{eq:markov}) we note that for any
$(x,y) \in A \times B$, $Pr(X=x,Y=y)=\frac{1}{|A||B|}$, and expand
$I(X+Y;Y)$ in two different ways 
\beqa
I(X+Y;Y) &=& H(X+Y)-H(X+Y|Y) = H(X+Y)-H(X) = H(X+Y) - \log_2 |A| \nonumber\\
&\le& \epsilon \log_2 |A| \nonumber\\
I(X+Y;Y) &=& H(Y) - H(Y|X+Y) = \log_2 |B| - E\left[ -\log_2
  p_{Y|X+Y}(Y|X+Y) \right] \nonumber\\ 
&=& \log_2 |B| - E\left[ -\log_2 \frac{p_{Y,X+Y}(Y,X+Y)}{p_{X+Y}(X+Y)}
  \right] \nonumber\\ 
&=& \log_2 |B| - E\left[ -\log_2 \frac{1}{|A||B|} \frac{|A||B|}{|T(X+Y)|} \right] \nonumber\\
&=& E\left[\log_2 \frac{|B|}{|T(X+Y)|}\right]\nonumber
\eeqa
to obtain 
\beq
E\left[\log_2 \frac{|B|}{|T(X+Y)|}\right] \le \epsilon \log_2 |A|.
\label{eq:ubexp}
\eeq
From (\ref{eq:F}), (\ref{eq:markov}) and (\ref{eq:ubexp}) we obtain:
\[
|F| \ge |A||B| \left(1- \frac{1}{c}\right).
\]
\end{proof}

Finally, we shall also rely on the following important theorem
from additive combinatorics, as
stated in~\cite{TaoVu06}, relating partial sum sets to full sum sets.
\begin{theorem}[Balog-Szemer\'{e}di-Gowers theorem]
\label{thm:balog}
Let $ A \subseteq G $ and $ B \subseteq G $ for some abelian group $ G
$ and let $ F \subseteq A \times B $ be such that 
\[
 |F| \geq |A||B|/K \text{ and } |A\stackrel{F}{+}B| \leq
 K'|A|^{1/2}|B|^{1/2} 
\] for some $ K
\geq 1 $ and $ K' > 0 $.  Then there exists $ A' \subseteq A, B'
\subseteq B $ such that 
\begin{align*}
|A'| &\geq \frac{|A|}{4\sqrt{2}K} \\
|B'| &\geq \frac{|B|}{4K} \\
|A'+B'| &\leq 2^{12}K^5(K')^3|A|^{1/2}|B|^{1/2}.
\end{align*}
\end{theorem}
Theorem~\ref{thm:balog} is proved in Chapter~6 of~\cite{TaoVu06}.

\subsection{Main lemma}
\label{sec:mainlemma}
In this subsection, we state and prove the main lemma at the core of 
our proof of Theorem~\ref{thm:rational}.
\begin{lemma}
\label{lem:mainlemma}
Let $p,q \in \mathbb{Z}$, $p,q \ne 0$, and
\[
\tilde{H} = \big[\tilde{h}_{ij}\big]=\left[
\begin{array}{ccc} 
1 & 0 & 0 \\ 
1 & p & 0 \\ 
1 & q & 1 
\end{array} \right],
\]
with the corresponding GIFC depicted in Figure~\ref{fig:3x3channel}.
Then $DoF(\tilde{H}) \le \frac{3}{2} - \epsilon(p,q)$, with
$\epsilon(p,q) > 0$.  In particular, this holds for
\[
\epsilon(p,q) = \frac{1}{12d(p,q)+2},
\]
where $ d(p,q) $ is as in Lemma~\ref{lem:setsum}.

\begin{figure}[t]
\centerline{\includegraphics[width=1.5in]{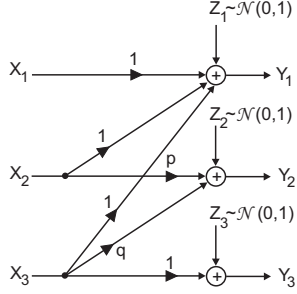}}
\caption{\footnotesize A three-user Gaussian IFC with channel matrix $\tilde{H}$. \label{fig:3x3channel} }
\end{figure}

\end{lemma}
\begin{proof}
We start by extending the definition of degrees-of-freedom to deterministic interference channels. Consider a $K$-user deterministic interference channel with input and output alphabets $\{\mc{X}_i\}_{i=1}^K$,$\{\mc{Y}_i\}_{i=1}^K$, defined by
\[
y_{i}(t)=v_i(x_1(t),\ldots,x_K(t)) \text{ , } i=1,\ldots,K; t=1,2,\ldots,
\]
where for each $i$, $x_i(t) \in \mc{X}_i$ must satisfy an average
power constraint $\sum_{t=1}^n x_i^2(t) \le n P$, and
$v_i:\mc{X}_1\times \cdots \times \mc{X}_K \to \mc{Y}_i$ is a
deterministic function. Let $\mc{C}(P)$ be the capacity region of the
channel with power constraint $P$. We define the degrees-of-freedom of
the deterministic channel by 
\[
DoF \dfn \limsup_{P\to \infty} \frac{\max_{\mb{R}\in\mc{C}(P)} \mb{1}^t\mb{R}}{(1/2) \log_2 P}.
\]

Due to Lemma \ref{lemma:deterministic}, the degrees-of-freedom of the GIFC with channel matrix $\tilde{H}$ is upper bounded by the degrees-of-freedom of the deterministic channel
\beq
y_i(t) = \sum_{j=1}^3 \tilde{h}_{ji} x_j(t) \text{ , } i=1,2,3,
\label{eq:detch}
\eeq 
where $x_i(t), y_i(t) \in \mathbb{Z}$,  and where the channel inputs are subject to the average power constraint $\frac{1}{n} \sum_{t=1}^n x_i^2(t)\le P$, for $i=1,2,3$.
We will prove by contradiction that the degrees-of-freedom of this deterministic channel is strictly smaller than $3/2$. Therefore, to get a contradiction, we assume that as $P$ goes to infinity, there are achievable rates  $(R_1(P),R_2(P),R_3(P))$ satisfying 
\[
\limsup_{P\to \infty}\frac{R_1(P)+R_2(P)+R_3(P)}{(1/2) \log_2 P} = \frac{3}{2}.
\]

This implies that there exists an increasing sequence of power
constraints $\{P_m\}_{m=1}^\infty$ with $\lim_{m\to\infty}P_m $ $ = \infty$, such that
\beq
\label{eq:dof3/2}
\lim_{m \to \infty}\frac{R_1(P_m)+R_2(P_m)+R_3(P_m)}{(1/2) \log_2 P_m} = \frac{3}{2}.
\eeq

For a power constraint $P_m$, consider a set of codebooks
$\{\mc{X}_{1,n,m},\mc{X}_{2,n,m},\mc{X}_{3,n,m}\}$ of block-length $n$
with rates $\{R_{1,m},R_{2,m},R_{3,m}\}$ in the deterministic channel
(\ref{eq:detch}), and with decoding functions $\{g_{1,n,m},g_{2,n,m},
g_{3,n,m}\}$ that achieve average error probabilities $\{P_{e,1,n,m},
P_{e,2,n,m}, P_{e,3,n,m}\}$.  We assume that
$\{R_{1,m},R_{2,m},R_{3,m}\}$ satisfy (\ref{eq:dof3/2}) as $m\to
\infty$ and that $\{P_{e,1,n,m}, P_{e,2,n,m}, P_{e,3,n,m}\}$ go to $0$
as $n\to \infty$ for fixed $m$. In addition, let $\bx_i^n$, $i=1,2,3$,
be the independent random vectors induced by the codebooks resulting
from the uniform distribution of the messages. It follows that
$\mc{X}_{i,n,m}\in \intgrs^n$ is the support set of $\bx_i^n$,
$i=1,2,3$. Due to Lemma \ref{lemma:repeated_codewords} we can assume
without loss of generality that the codebooks do not have repeated
codewords. This implies that each $\bx_i^n$ is chosen uniformly in
$\mc{X}_{i,n,m}$, or equivalently that
$Pr(\bx_i^n)=\frac{1}{|\mc{X}_{i,n,m}|}$, for $\bx_i^n \in
\mc{X}_{i,n,m}$, $i=1, 2, 3$.  

Using Fano's inequality we write (where, for simplicity, we suppress
the dependence on $ m $ of the variables $ \bx_1, \bx_2, $ etc.).
\beqa
n(R_{1,m}+R_{2,m}+R_{3,m}-\delta_n) &\le& \sum_{i=1}^3 I(\bx_i^n;\by_i^n)\nonumber\\
&=& \sum_{i=1}^3 \big[ H(\by_i^n)-H(\by_i^n|\bx_i^n) \big]\nonumber\\
&=& H(\bx_1^n+\bx_2^n+\bx_3^n) - H(\bx_2^n+\bx_3^n) + H(p \bx_2^n + q
\bx_3^n) \nonumber \\ & & - H(q \bx_3^n) + H(\bx_3^n) \nonumber\\
&=& H(\bx_1^n+\bx_2^n+\bx_3^n) - H(\bx_2^n+\bx_3^n) + H(p \bx_2^n + q \bx_3^n). \label{eq:sument1}
\eeqa
To handle the first term of (\ref{eq:sument1}), we use the following lemma, which follows from exercise 8.7 of \cite{CT06} and Jensen's inequality.
\begin{lemma}
\label{lemma:gub}
Let $\mb{X}^n=(X_1,\ldots, X_n)$ be a discrete random vector on $\mathbb{Z}^n$. Then,
\[
H(\mb{X}^n) \le \frac{n}{2} \log_2\left[2 \pi e \left(\frac{1}{n} \sum_{i=1}^n Var(X_i) +\frac{1}{12}\right) \right],
\]
\end{lemma}
where $ Var(X_i) = E(X_i^2) - E^2(X_i) $ is the variance of $ X_i $.
Therefore, using the independence among the input signals and the fact
that $ \sum_{t=1}^n Var(x_{i,t}) \leq E\|\bx_i^n\|_2^2 \le n P_m$ we obtain
\beq
H(\bx_1^n+\bx_2^n+\bx_3^n) \le \frac{n}{2} \log\left[2 \pi e \left(3
  P_m+\frac{1}{12} \right)\right]. 
\label{eq:hsum1}
\eeq

The remaining two terms in (\ref{eq:sument1}) will be bounded in two different ways. First we use the following simple bounds on $H(a X + b Y)$, valid for any $a,b \ne 0$ and independent random variables $X$ and $Y$:
\beqa
H(a X+b Y) &\le& H(a X + b Y, a X) = H(a X) + H(a X + b Y | a X) = H(X) + H(Y)  \nonumber\\
H(a X+b Y) &\ge& H(a X+ b Y| b Y) = H(a X) = H(X).
\eeqa
Using these bounds we get
\beqa
H(p \bx_2^n + q \bx_3^n) - H(\bx_2^n+\bx_3^n) &\le& H(\bx_2^n)+H(\bx_3^n) - \max\big\{H(\bx_2^n);H(\bx_3^n) \big\} \nonumber\\
&=& \min\big\{H(\bx_2^n);H(\bx_3^n)\big\}\nonumber\\
&=& \log_2 \min \{|\mc{X}_{2,n,m}|;|\mc{X}_{3,n,m}| \}.
\label{eq:sument2}
\eeqa

We define $f_{min}(n,m)$ such that $\log_2 \min
\{|\mc{X}_{2,n,m}|;|\mc{X}_{3,n,m}| \}= (\frac{1}{4}+f_{min}(n,m)) n
\log_2 P_m$. Note that by the assumed dependence of $ \mc{X}_{i,n,m} $
on $ R_i $, $ \lim_{n\rightarrow \infty} f_{min}(n,m) $ exists for
  each $ m $.
Then, using (\ref{eq:dof3/2}), (\ref{eq:sument1}),
(\ref{eq:hsum1}), and (\ref{eq:sument2}), we have 
\beqa
\frac{3}{2} &=& \lim_{m\to \infty} \frac{\lim_{n\to \infty} (R_{1,m}+R_{2,m}+R_{3,m} - \delta_n)}{(1/2)\log_2 P_m} \nonumber\\
&\le& 1+\frac{1}{2} + \liminf_{m\to \infty} \frac{\lim_{n\to \infty} f_{min}(n,m) \log_2 P_m}{(1/2)\log_2P_m}\nonumber
\eeqa
which implies
\beq
\liminf_{m\to \infty} \lim_{n \to \infty} f_{min}(n,m) \ge 0.
\label{eq:liminf}
\eeq

As a second option, we use Lemma \ref{lemma:gub} to get the upper bound
\beq
H(p \bx_2^n + q \bx_3^n) \le \frac{n}{2} \log_2 \left[2 \pi e \left((p^2+q^2) P_m +\frac{1}{12}\right) \right]
\label{eq:hsum2}
\eeq
which we use to get
\beqa
H(p \bx_2^n + q \bx_3^n) - H(\bx_2^n+\bx_3^n) &\le& \frac{n}{2} \log_2 \left[2 \pi e \left((p^2+q^2) P_m +\frac{1}{12}\right) \right]
- \log_2 \max\big\{|\mc{X}_{2,n,m}|;|\mc{X}_{3,n,m}| \big\}. \nonumber\\
\label{eq:sument3}
\eeqa

We define $f_{max}(n,m)$ such that $\log_2 \max
\{|\mc{X}_{2,n,m}|;|\mc{X}_{3,n,m}| \}= (\frac{1}{4}+f_{max}(n,m)) n
\log_2 P_m$. Note, as above, that $ \lim_{n\rightarrow \infty}
f_{max}(n,m) $ exists for   each $ m $.
Then, using (\ref{eq:dof3/2}), (\ref{eq:sument1}),
(\ref{eq:hsum1}), and (\ref{eq:sument3}), we have 
\beqa
\frac{3}{2} &=& \lim_{m\to \infty} \frac{\lim_{n\to \infty} (R_{1,m}+R_{2,m}+R_{3,m} - \delta_n)}{(1/2)\log_2 P_m} \nonumber\\
&\le& 1+1 -\frac{1}{2} - \limsup_{m\to \infty} \frac{\lim_{n\to \infty} f_{max}(n,m) \log_2 P_m}{(1/2)\log_2 P_m} \nonumber
\eeqa
which implies
\beq
\limsup_{m\to \infty} \lim_{n\to \infty} f_{max}(n,m) \le 0.
\label{eq:limsup}
\eeq

Since $f_{max}(n,m) \ge f_{min}(n,m)$, (\ref{eq:liminf}) and (\ref{eq:limsup}) imply
\beq
\lim_{m\to \infty} \lim_{n\to \infty} f_{max}(n,m)= \lim_{m\to \infty} \lim_{n\to \infty} f_{min}(n,m)= 0
\label{eq:limit}
\eeq
and, as a result, for any $\xi > 0$, there exist $m_0(\xi)$ and $n_0(\xi,m)$ such that, $|f_{min}(n,m)| \le \xi $ and $|f_{max}(n,m)| \le \xi $ for all $m \ge m_0(\xi)$ and $n\ge n_0(\xi, m)$. Therefore, for any $\xi >0$, $m \ge m_0(\xi)$ and $n\ge n_0(\xi, m)$ we have
\begin{align}
\left(\frac{1}{4}-\xi\right) n\log_2 P_m \le \log_2&\left(\min\{|\mc{X}_{2,n,m}|;|\mc{X}_{3,n,m}| \}\right)\le\nonumber\\
& \le \log_2\left(\max\{|\mc{X}_{2,n,m}|;|\mc{X}_{3,n,m}| \}\right) \le \left(\frac{1}{4}+\xi\right) n\log_2 P_m.
\label{eq:cardinalitybounds}
\end{align}

We define $g(n,m)\ge 0$ such that $H(\bx_2^n+\bx_3^n)=(1+g(n,m))\log_2(\max\{|\mc{X}_{2,n,m}|;|\mc{X}_{3,n,m}|\})$.
Then, using (\ref{eq:dof3/2}), (\ref{eq:sument1}), (\ref{eq:hsum2}), and (\ref{eq:limit}) we have
\beqa
\frac{3}{2} &=& \lim_{m\to \infty} \frac{\lim_{n\to \infty} (R_{1,m}+R_{2,m}+R_{3,m} - \delta_n)}{(1/2)\log_2 P_m} \nonumber\\
&\le& 1+1-\frac{1}{2}-0 
-\limsup_{m\to \infty} \frac{\limsup_{n \rightarrow \infty} g(n,m)((1/4)+f_{max}(n,m))\log_2 P_m}{(1/2)\log_2P_m}
\eeqa
which, together with (\ref{eq:limit}) and the condition $g(n,m)\ge 0$, imply
\beq
\lim_{m\to \infty} \limsup_{n\to \infty} g(n,m) = 0
\eeq
and, as a result, for any $\epsilon > 0$, there exist $m_0(\epsilon)$ and $n_0(\epsilon,m)$ such that, $0\le g(n,m) \le \epsilon $  for all $m \ge m_0(\epsilon)$ and $n\ge n_0(\epsilon, m)$. 

Therefore, for any $\epsilon > 0$, $m \ge m_0(\epsilon)$ and $n\ge n_0(\epsilon, m)$ we have 
\[
H(\bx_2^n + \bx_3^n) \le (1+\epsilon) \log_2
\max\{|\mc{X}_{2,n,m}|,|\mc{X}_{3,n,m}| \}.
\]
For any given $c > 1$,
Lemma \ref{lemma:exG} guarantees the existence of $F_{c,n,m} \subseteq
\mc{X}_{2,n,m}\times \mc{X}_{3,n,m}$ such that 
\[
|F_{c,n,m}| \ge \frac{|\mc{X}_{2,n,m}||\mc{X}_{3,n,m}|}{K} \text{ and } \bigg|\mc{X}_{2,n,m}\stackrel{F_{c,n,m}}{+} \mc{X}_{3,n,m}\bigg| \le K'_{n,m} |\mc{X}_{2,n,m}|^{1/2}|\mc{X}_{3,n,m}|^{1/2}
\]
with $K=c/(c-1)$ and $K_{n,m}'=(\max\{|\mc{X}_{2,n,m}|;|\mc{X}_{3,n,m}|\})^{1/2+c \epsilon} (\min\{|\mc{X}_{2,n,m}|;|\mc{X}_{3,n,m}|\})^{-1/2}$. 
Using Theorem \ref{thm:balog} with $F_{c,n,m}$, it follows that there exist $\mc{X}_{2,n,m}'$
$\subseteq \mc{X}_{2,n,m}$, and $\mc{X}_{3,n,m}' \subseteq \mc{X}_{3,n,m}$ such that
\beqa
|\mc{X}_{2,n,m}'| &\ge& \frac{|\mc{X}_{2,n,m}|}{4 \sqrt{2} K} \label{eq:x2'}\\
|\mc{X}_{3,n,m}'| &\ge& \frac{|\mc{X}_{3,n,m}|}{4 \sqrt{2} K} \label{eq:x3'}\\
|\mc{X}_{2,n,m}'+\mc{X}_{3,n,m}'| &\le& 2^{12} K^5 (K'_{n,m})^3 |\mc{X}_{2,n,m}|^{1/2}|\mc{X}_{3,n,m}|^{1/2}. \label{eq:x2'x3'}
\eeqa

From (\ref{eq:x2'}) and (\ref{eq:x3'}) and Lemma
\ref{lemma:pruned_codebooks} it follows that the set of codebooks
$\{\mc{X}_{1,n,m}, \mc{X}_{2,n,m}', \mc{X}_{3,n,m}'\}$ has rates
$(R_{1,m}, R_{2,m}-\frac{1}{n}\log_2(4\sqrt{2}K), R_{3,m}-\frac{1}{n}\log_2(4\sqrt{2}K))$, 
and using the decoding functions $\{g_{1,n,m}, g_{2,n,m}, g_{3,n,m}\}$
the average error probabilities are no larger than $\{P_{e,1,n,m}*32
K^2, P_{e,2,n,m}*32 K^2, P_{e,3,n,m}*32 K^2\}$. Let
$\btx_1^n, \btx_2^{n}, \btx_3^{n}$ be the random vectors induced by the
codebooks $\{\mc{X}_{1,n,m}, $ $ \mc{X}_{2,n,m}', $ $
\mc{X}_{3,n,m}'\}$. Using Fano's inequality and absorbing in
$\delta_n'$ all the constants that vanish with $n\to \infty$ we write: 
\begin{align}
n(R_{1,m}+R_{2,m}+R_{3,m}-\delta_n') \le& \sum_{i=1}^3 I(\btx_i^n;\by_i^n)\nonumber\\
=& \sum_{i=1}^3 \big[ H(\by_i^n)-H(\by_i^n|\btx_i^n) \big]\nonumber\\
=& H(\btx_1^n+\btx_2^n+\btx_3^n) - H(\btx_2^n+\btx_3^n) + H(p \btx_2^n + q \btx_3^n) - H(q \btx_3^n) + H(\btx_3^n) \nonumber\\
=& H(\btx_1^n+\btx_2^n+\btx_3^n) - H(\btx_2^n+\btx_3^n) + H(p \btx_2^n + q \btx_3^n). \label{eq:sument4}
\end{align}
To bound the first term of (\ref{eq:sument4}), as before, we use Lemma \ref{lemma:gub}, obtaining
\beq
H(\btx_1^n+\btx_2^n+\btx_3^n) \le \frac{n}{2} \log\left[2 \pi e \left(3 P_m+\frac{1}{12} \right)\right].
\label{eq:hsum3}
\eeq

To bound the remaining two terms of (\ref{eq:sument4}) we use techniques from additive combinatorics to upper bound the cardinality of the support set of $(p \btx_2^n + q \btx_3^n)$ in terms of the cardinality of the support set of $(\btx_2^n+\btx_3^n)$.  

From Lemma \ref{lem:setsum}, (\ref{eq:x2'}), (\ref{eq:x3'}) and (\ref{eq:x2'x3'}) we have 
\beqa
|p\cdot \mc{X}_{2,n,m}' + q \cdot \mc{X}_{3,n,m}'| &\le& | \mc{X}_{2,n,m}' + \mc{X}_{3,n,m}'|^{d(p,q)} |\mc{X}_{2,n,m}'|^{[1-d(p,q)]/2} |\mc{X}_{3,n,m}'|^{[1-d(p,q)]/2}\nonumber\\
&\le& \left\{ 2^{12} K^5 (K'_{n,m})^3 |\mc{X}_{2,n,m}|^{1/2}|\mc{X}_{3,n,m}|^{1/2}\right\}^{d(p,q)} \nonumber\\
&& \cdot\left(\frac{|\mc{X}_{2,n,m}|}{4 \sqrt{2} K}\right)^{[1-d(p,q)]/2} \left(\frac{|\mc{X}_{3,n,m}|}{4 \sqrt{2} K}\right)^{[1-d(p,q)]/2}\nonumber\\
\eeqa
which together with the bound $H(X) \le \log_2|\mc{X}|$ results in
\beqa
H(p \btx_2^n + q \btx_3^n) &\le& \frac{1}{2} \log_2|\mc{X}_{2,n,m}|+\frac{1}{2} \log_2|\mc{X}_{3,n,m}|+3 d(p,q) \log_2 K'_{n,m} + \tilde{K}_{c,p,q} \nonumber\\
&=& \frac{1}{2} \log_2|\mc{X}_{2,n,m}|+\frac{1}{2}
\log_2|\mc{X}_{3,n,m}|+\frac{3}{2} d(p,q)\left(1+2 c \epsilon\right)
\log_2
\left(\max\{|\mc{X}_{2,n,m}|;|\mc{X}_{3,n,m}|\}\right)\nonumber\\ 
& &-\frac{3}{2} d(p,q) \log_2 \left(\min\{|\mc{X}_{2,n,m}|;|\mc{X}_{3,n,m}|\}\right) + \tilde{K}_{c,p,q} 
\eeqa
where $\tilde{K}_{c,p,q}$ is some constant independent of $n$ and $ m $.

On the other hand using (\ref{eq:x2'}) and (\ref{eq:x3'}) we have
\beqa
H(\btx_2^n+ \btx_3^n) &\ge& \max\{H(\btx_2^n);H(\btx_3^n)\} \nonumber\\
&=& \log_2 \left(\max\{|\mc{X}_{2,n,m}'|;|\mc{X}_{3,n,m}'| \}\right) \nonumber\\
&\ge& \log_2 \left(
\max\{|\mc{X}_{2,n,m}|;|\mc{X}_{3,n,m}|\}\right)-\log_2\left(4
\sqrt{2} K\right) \label{eq:Hsumlb}.
\eeqa
Therefore
\beqa
H(p \btx_2^n + q \btx_3^n)-H(\btx_2^n+ \btx_3^n) &\le& \frac{1}{2} \log_2\left(\min\{|\mc{X}_{2,n,m}|;|\mc{X}_{3,n,m}|\}\right)-\frac{1}{2} \log_2\left(\max\{|\mc{X}_{2,n,m}|;|\mc{X}_{3,n,m}|\}\right)\nonumber\\
& &+\frac{3}{2} d(p,q)\left(1+2 c \epsilon\right) \log_2 \left(\max\{|\mc{X}_{2,n,m}|;|\mc{X}_{3,n,m}|\}\right)\nonumber\\
& &-\frac{3}{2} d(p,q) \log_2 \left(\min\{|\mc{X}_{2,n,m}|;|\mc{X}_{3,n,m}|\}\right) + \tilde{K}_{c,p,q}' \nonumber\\
&\le& 0 + \frac{3}{2} d(p,q)\left(1+2 c \epsilon\right) \log_2 \left(\max\{|\mc{X}_{2,n,m}|;|\mc{X}_{3,n,m}|\}\right)\nonumber\\
& &-\frac{3}{2} d(p,q) \log_2 \left(\min\{|\mc{X}_{2,n,m}|;|\mc{X}_{3,n,m}|\}\right) + \tilde{K}_{c,p,q}' 
\eeqa
where $\tilde{K}_{c,p,q}'$ is some other constant independent of $n$
and $ m $. Using the bounds in (\ref{eq:cardinalitybounds}) we obtain:
\beqa
H(p \btx_2^n + q \btx_3^n)-H(\btx_2^n+ \btx_3^n) &\le& \left[3 \xi d(p,q)+3 c \epsilon d(p,q)\left(\frac{1}{4}+\xi\right)\right] n\log_2 P_m+ \tilde{K}_{c,p,q}',
\eeqa
which, together with (\ref{eq:sument4}) and (\ref{eq:hsum3}), imply
\begin{multline}
\frac{R_{1,m}+R_{2,m}+R_{3,m}}{(1/2)\log_2(P_m)} \le \\ \frac{\frac{1}{2}
  \log_2\left[2 \pi e \left(3 P_m+\frac{1}{12} \right)\right]+\left[3
    \xi d(p,q)+3 c \epsilon d(p,q)\left(\frac{1}{4}+\xi\right)\right]
  \log_2 P_m+ \frac{\tilde{K}_{c,p,q}'}{n}+\delta_n'}{(1/2)\log_2 P_m} 
\end{multline}
which is strictly smaller than $(3/2)$ for small enough $\xi$ and
$\epsilon$, and large enough $m$ and $n$. This contradicts
(\ref{eq:dof3/2}). 

We can refine the above analysis to find
the smallest $ \epsilon(p,q) $ such that assuming that $ DoF(\tilde{H})
= 3/2-\epsilon(p,q) $ does not lead to a contradiction.  The 
expression for $ \epsilon(p,q) $ in the statement of the
lemma is obtained by avoiding bounds on $ H(\bx^n_2+\bx^n_3) $ until the
last step in the analysis.  For example, Lemma~\ref{lemma:exG} is
applied with 
\[ H(\bx^n_2+\bx^n_3) = \left(\frac{H(\bx^n_2+\bx^n_3)}{\log_2
(\max\{|\mc{X}_{2,n,m}|;|\mc{X}_{3,n,m}|\})}\right)\log_2
(\max\{|\mc{X}_{2,n,m}|;|\mc{X}_{3,n,m}|\}). \]
Note that $ H(\bx^n_2+\bx^n_3) \neq H(\btx^n_2+\btx^n_3) $
and we still rely on the lower bound~(\ref{eq:Hsumlb}) for the latter entropy.
The entropy $ H(\bx^n_2+\bx^n_3) $ and
cardinalities in the final expression are bounded by assuming that $
DoF(\tilde{H}) = 3/2-\epsilon(p,q) $, which can easily be shown, following
steps similar to those leading to~(\ref{eq:cardinalitybounds}),
to imply
\[
 \frac{H(\bx^n_2+\bx^n_3)}{n\log_2 P} \leq \left(\frac{1}{4}+\frac{\epsilon(p,q)}{2}\right)+o(1), \]
 and
\begin{multline*}
\left(\frac{1}{4}-\frac{\epsilon(p,q)}{2}\right)+o(1)
\leq \frac{\log_2 (
\min\{|\mc{X}_{2,n,m}|;|\mc{X}_{3,n,m}|\})}{n\log_2 P} 
\leq  \\ 
\frac{\log_2 (
\max\{|\mc{X}_{2,n,m}|;|\mc{X}_{3,n,m}|\})}{n\log_2 P} 
\leq 
\left(\frac{1}{4}+\frac{\epsilon(p,q)}{2}\right)+o(1).
\end{multline*}
The details are omitted.
\end{proof}

\subsection{Proof of Theorem~\ref{thm:rational}}
\label{sec:genK}
By assumption all the entries of $H$ are non-zero and rational, and therefore, there exists a diagonal matrix $D_r$ with positive diagonal entries such that $\bar{H}\dfn H D_r$ has non-zero integer entries. From Lemma \ref{lem:invariance} it follows that 
\beq
\label{eq:dofeq}
DoF(H)=DoF(\bar{H}). 
\eeq

Consider the channel formed by the transmitters and receivers of users
$i$, $j$ and $k$. Let $\bar{H}_{i,j,k}\in \intgrs^{3 \times 3}$ be
the principal minor (matrix) of 
$\bar{H}$ corresponding to the $(i,j,k)$-th rows and columns. 
Due to the independence of the signals of the different 
users it follows that 
\beq
\label{eq:sum3r1}
\max_{(R_1,\ldots,R_K) \in \mc{C}(\bar{H},\mb{1},P \mb{1})} R_i+R_j+R_k= \max_{(R_i,R_j,R_k) \in \mc{C}(\bar{H}_{i,j,k},\mb{1},P \mb{1})} R_i+R_j+R_k,
\eeq
i.e. the other users cannot help users $i$, $j$, and $k$ to improve
their rates. In addition, since interference cannot help a given
receiver in decoding the signal of interest it follows that we can set
some of the cross-gains in  $\bar{H}_{i,j,k}$ to zero without reducing
the maximum achievable sum rate. More specifically, defining
$\hat{H}_{i,j,k}=[\hat{h}_{i,j,k}(m,n)]$ by $\hat{h}_{i,j,k}(m,n) \dfn
\bar{h}_{i,j,k}(m,n)\cdot 1(m\ge n)$ we have 
\beq
\label{eq:sum3r2}
\max_{(R_i,R_j,R_k) \in \mc{C}(\bar{H}_{i,j,k},\mb{1},P \mb{1})} R_i+R_j+R_k \le \max_{(R_i,R_j,R_k) \in \mc{C}(\hat{H}_{i,j,k},\mb{1},P \mb{1})} R_i+R_j+R_k.
\eeq
Furthermore, it is easy to see\footnote{Letting $\hat{H}_{i,j,k}=[a,
    0, 0; b, c, 0; d, e, f]$ in Matlab matrix notation, we can choose
  $\hat{D}_t=[bd, 0, 
    0; 0, ad, 0; 0, 0, ab]$ and $\hat{D}_r=[1/(abd), 0, 0; 0,
    1/a, 0; 0, 0, 1/(abf)]$.} that there exist diagonal matrices
$\hat{D}_t$, $\hat{D}_r$ with positive diagonal entries such that
$\hat{D}_t \hat{H}_{i,j,k} \hat{D}_r=\tilde{H}_{i,j,k}$ where 
\[
\tilde{H}_{i,j,k} = \left[
\begin{array}{ccc} 
1 & 0 & 0 \\ 
1 & p_{i,j,k} & 0 \\ 
1 & q_{i,j,k} & 1 
\end{array} \right]
\] 
for some $p_{i,j,k},q_{i,j,k} \in \intgrs$, $p_{i,j,k},q_{i,j,k} \ne 0$. Using (\ref{eq:sum3r2}), Lemma~\ref{lem:invariance} and Lemma~\ref{lem:mainlemma} we have
\beq
\label{eq:dofineq}
DoF(\bar{H}_{i,j,k})\le DoF(\hat{H}_{i,j,k})= DoF(\tilde{H}_{i,j,k})\le \frac{3}{2}-\epsilon(p_{i,j,k},q_{i,j,k})
\eeq
where $\epsilon(p_{i,j,k},q_{i,j,k}) > 0$. 

Considering every possible subset of users $\{i,j,k\}\subseteq\{1,\ldots,K\}$ and adding the corresponding sum rates, the rate of each user appears ${K-1 \choose 2}$ times in the sum. Therefore, we have
\beqa
{K-1 \choose 2} DoF(H)&\stackrel{(a)}{=}& {K-1 \choose 2} DoF(\bar{H})\nonumber\\
&=&\limsup_{P \to \infty}\frac{\max_{(R_1,\ldots,R_K)\in\mc{C}(\bar{H},\mb{1},P \mb{1})} {K-1 \choose 2}\sum_{i=1}^K R_i}{\frac{1}{2} \log_2 P} \nonumber\\
&\le& \limsup_{P \to \infty}\frac{\sum_{\{i,j,k\}\subseteq \{1,\ldots,K\}}\max_{(R_1,\ldots,R_K)\in\mc{C}(\bar{H},\mb{1},P \mb{1})}  \left(R_i+R_j+R_k\right)}{\frac{1}{2}\log_2 P} \nonumber\\
&\stackrel{(b)}{=}& \limsup_{P \to \infty}\frac{\sum_{\{i,j,k\}\subseteq \{1,\ldots,K\}}\max_{(R_i,R_j,R_k)\in\mc{C}(\bar{H}_{i,j,k},\mb{1},P \mb{1})}  \left(R_i+R_j+R_k\right)}{\frac{1}{2}\log_2 P} \nonumber\\
&\le& \sum_{\{i,j,k\}\subseteq \{1,\ldots,K\}}DoF(\bar{H}_{i,j,k}) \nonumber\\
&\stackrel{(c)}{\le}& {K \choose 3}\left( \frac{3}{2} - \delta \right),
\eeqa
where (a) is due to (\ref{eq:dofeq}), (b) follows from (\ref{eq:sum3r1}), (c) is obtained from (\ref{eq:dofineq}), and where we defined $\delta \dfn \min_{\{i,j,k\}\subseteq \{1,\ldots,K\}} \epsilon(p_{i,j,k},q_{i,j,k}) > 0$. We finally obtain
\beq
DoF(H) \le \frac{K}{2}-\frac{K}{3}\delta < \frac{K}{2},
\eeq
establishing the theorem. 

\section{A 3-user rational GIFC example}
\label{sec:example}
\begin{figure}[t]
\centerline{\includegraphics[width=1.5in]{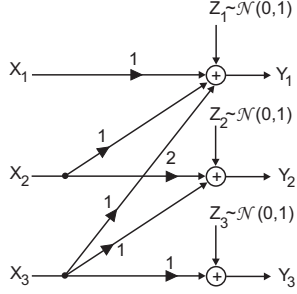}}
\caption{\footnotesize Three-user GIFC of the example in
  Section~\ref{sec:example}. \label{fig:3x3channel2}}
\end{figure}

In this section, we will derive lower and upper bounds on the degrees-of-freedom of the 3-user GIFC with channel matrix
\[
H = \left[
\begin{array}{ccc} 
1 & 0 & 0 \\ 
1 & 2 & 0 \\ 
1 & 1 & 1 
\end{array} \right]
\]
represented in Figure~\ref{fig:3x3channel2}. This channel is a special
case of the one considered in Lemma~\ref{lem:mainlemma}, with $p=2$ and
$q=1$. From this lemma, we obtain: 
\beq
\label{eq:dofexample}
DoF(H) \le \frac{3}{2}-\frac{1}{12 d(2,1) + 2}
\eeq
where $d(p,q)=2 \max\{|p|,|q|\} + 5$ was obtained in Lemma~\ref{lem:setsum}. For the special case of $q=1$ the result of Lemma~\ref{lem:setsum} can be easily strengthened to get $d(p,1)=2 |p| + 3$.
Evaluating (\ref{eq:dofexample}) with $d(2,1)=7$ we obtain $DoF(H) \le
1.4884$.   It should be possible to improve this bound in a number of
ways, such as by
improving on $ d(2,1) = 7 $, and possibly by improving on the power $ 3
$ in the term $
(K')^{3} $ appearing in Theorem~\ref{thm:balog}.  Another possibility
might be to forgo this theorem for a different approach, such as one based on
Exercise~2.5.4 of~\cite{TaoVu06}.

To get a lower bound on $DoF(H)$, we will describe a communication scheme that aims to achieve good interference alignment at receiver 1 by aligning the interfering signals of transmitters 2 and 3, while achieving good separation at receiver 2 between the signal of transmitter 2 and the interference of transmitter 3.

As was done in \cite{CJS07}, we design the communication scheme for a deterministic interference channel and later show how to extend the scheme to the Gaussian channel. We derive the deterministic channel from the Gaussian IFC by removing the Gaussian noise and constraining the inputs to be integers. Let $A_1 = \{0,1\}$, $A_2 = \{0, 2, 4\}$, $A_3 = \{0,2\}$, and $Q=8$. We communicate information independently over $L$ levels, without coding over time. The signal of user $i$ at time $t$ is given by:
\beq
x_i(t) = \sum_{\ell=1}^L m_{i,\ell}(t) Q^{\ell-1},
\eeq
where $m_{i,\ell}(t)\in A_i$ is the message of user $i$ in level $\ell$ at time $t$. 

Since $A_1+A_2+A_3=\{0,1,2,3,4,5,6,7\}$, the signal at receiver 1 can be written as  
\[
y_1(t)=\sum_{i=1}^3 x_i(t) = \sum_{\ell=1}^L w_{1,\ell}(t) Q^{\ell-1}
\]
with $w_{1,\ell}(t)=m_{1,\ell}(t)+m_{2,\ell}(t)+m_{3,\ell}(t)$. Therefore, by computing the $Q$-ary decomposition of $y_1(t)$ we can recover the sums $m_{1,\ell}(t)+m_{2,\ell}(t)+m_{3,\ell}(t)$ at each of the $L$ levels. In addition, since $A_2+A_3 = \{0,2,4,6\}$ we have $m_{1,\ell}(t) = 1(w_{1,\ell}(t) \in \{1,3,5,7\})$, so we can directly determine $m_{1,\ell}(t)$ from $w_{1,\ell}(t)$.

Similarly, at receiver 2 we compute
\[
\frac{y_2(t)}{2} = x_2(t)+\frac{1}{2} x_3(t) = \sum_{\ell=1}^L w_{2,\ell}(t) Q^{\ell-1}
\]
with $w_{2,\ell}(t)=m_{2,\ell}(t)+(1/2) m_{3,\ell}(t) \in \{0,1,2,3,4,5\}$,
from which we can compute $m_{2,\ell}(t)=w_{2,\ell}(t)-[w_{2,\ell}(t) \mod
  2]$. 

Finally, receiver 3 can directly recover $m_{3,\ell}(t)$ at all levels from the received signal $y_3(t)=x_3(t)$.

To compute the achievable degrees-of-freedom of this scheme we note that since $|x_i(t)| < Q^L$ the transmission power at each transmitter is smaller than $Q^{2L}$. On the other hand the rate of users $1$ and $3$ is $L \log_2 2$ while the rate of user $2$ is $L \log_2 3$. Therefore, we obtain for the deterministic channel
\[
DoF \ge \frac{2 L \log_2 2 + L \log_2 3}{\frac{2L}{2} \log_2 8} = \frac{2+\log_2 3}{3} \approx 1.19499.
\]

We now informally argue that the same degrees-of-freedom can be
achieved in the Gaussian channel. We essentially use the same
multi-level coding scheme, but we now encode the signals of each level
over long blocks of time. The lower levels may be severely affected by
noise, but as the level $\ell$ increases, the influence of the noise
becomes smaller, ultimately being insignificant. As a result, the
amount of redundancy that needs to be added to the signal of level $\ell$
to ensure low probability of decoding error goes to 0 as $\ell$ grows to
infinity. It follows that for $\ell$ large enough, the achievable rates
in the Gaussian channel at level $\ell$ approach the achievable rates in
the deterministic channel, and since the rates of the lower levels do
not affect the degrees-of-freedom, we conclude that
$DoF(H)\ge\frac{2+\log_2 3}{3}$ (see \cite{CJS07} for a similar
argument).  

In summary, using the lower and upper bounds that we derived we have,
\[
1.19499 \le DoF(H) \le 1.4884.
\]

\begin{remark} The achievable scheme that we described is simple to analyze because there are no ``carry overs'' across the different levels, and the signals and interference are ``orthogonal'' in the sense that there is no need to code over time to ensure reliable decoding in the deterministic channel. In choosing the sets $A_1$, $A_2$ and $A_3$ we tried to obtain small $|A_2+A_3|$ to align the interference at receiver 1 and simultaneously obtain large $|2 A_2+ A_3|$ to achieve good signal-interference separation at receiver 2. With these design guidelines, one could optimize the sets $A_1$, $A_2$ and $A_3$ (with possibly larger $Q$) in order to improve the achievable degrees-of-freedom. In addition, Han-Kobayashi-type schemes \cite{HK81} at each level where part of the interference is decoded and subtracted may result in better performance than purely orthogonal schemes. 
\end{remark}

\section{Conclusion}
\label{sec:conclusion}
We have shown that
the degrees-of-freedom of $ K > 2 $ user, real, scalar GIFCs is 
sensitive to whether the channel gains have rational or irrational
values, and it is, in fact, discontinuous at all fully connected,
rational gain matrices (up to the invariance property of
Lemma~\ref{lem:invariance}).  Specifically,
Theorem~\ref{thm:irrational} shows that certain fully connected real,
scalar GIFCs with irrational, algebraic coefficients have degrees-of-freedom
exactly equal to the known upper bound of $ K/2 $, the first such
examples for real, scalar GIFCs.  Theorem~\ref{thm:rational}, on the
other hand, shows that if all coefficients are non-zero rationals, the
degrees-of-freedom is strictly bounded away from $ K/2 $, for $ K > 2$.  These
theorems are established by appealing to major results in
mathematics on the inapproximability of irrational, algebraic
numbers by rational numbers, in the case of
Theorem~\ref{thm:irrational}, and on the combinatorics of additive
sets, in the case of Theorem~\ref{thm:rational}.  In the latter case,
previously used information theoretic converse techniques, which are not
sensitive to the rationality of channel coefficients, do not suffice.
We believe these results may have some implications for real GIFCs under
channel parameter uncertainty, since in this case, channel
coefficients with irrational and rational coefficients would have to
be dealt with simultaneously.  Additionally, in practical systems,
computations for encoding and decoding are ultimately restricted to
finite precision, and hence rational numbers, suggesting that
additive combinatorics based bounds on achievable rates may have
practical relevance.  

Throughout this paper, we have been concerned with real, scalar GIFCs.  
Theorem~\ref{thm:irrational} can be readily extended to the complex and
vector cases, revealing an additional class of 
$K/2$ degrees-of-freedom vector GIFCs, complementing those already
known~\cite{CadJaf07}.  The extension of Theorem~\ref{thm:rational} to
complex and vector GIFCs seems less trivial.  For instance, the
example of a $ K/2 $ degrees-of-freedom achieving $K$-user
two-dimensional vector GIFC 
in~\cite{CadJaf07} 
actually has integer coefficients, though they are a very special choice.
Any extension of Theorem~\ref{thm:rational} would have to avoid
such special cases.  More significantly, our crucial Lemma~\ref{lem:mainlemma}
can be shown, using the interference alignment technique
of~\cite{CadJaf07}, to not hold in the complex or vector cases.
Nevertheless, we conjecture that for any fixed (complex) vector dimension,
limitations on the degrees-of-freedom similar to
Theorem~\ref{thm:rational} do exist for a sufficiently large
number of users $ K $.  As noted, establishing such a result would require
analyzing few-user GIFCs that are more complicated than the
$3$-user channel of Lemma~\ref{lem:mainlemma}.  It is likely that tools
from additive combinatorics will still prove useful, though they would
need to be applied differently from the proof of
Lemma~\ref{lem:mainlemma}.   The topic of characterizing the
degrees-of-freedom of rational vector GIFCs and
scalar, complex GIFCs having channel gains with rational real and
imaginary parts is left for future work.

\end{document}